\documentclass[letterpaper,11pt]{article}
\usepackage[margin=1in]{geometry}
\usepackage[utf8]{inputenc} 
\usepackage[T1]{fontenc}    
\usepackage[colorlinks=true,breaklinks=true,bookmarks=true,urlcolor=blue,citecolor=blue,linkcolor=blue,bookmarksopen=false,draft=false]{hyperref}
\usepackage{url}            
\usepackage{booktabs}       
\usepackage{multirow}
\usepackage{microtype}      
\usepackage{amsfonts,amsmath, amsthm,mathtools,stmaryrd}
\usepackage{xcolor}
\usepackage[inline]{enumitem}
\usepackage{appendix}
\usepackage{float, graphicx, subcaption}

\usepackage{natbib}

\newcommand{\XCal}{\mathcal{X}}
\newcommand{\YCal}{\mathcal{Y}}

\newcommand{\SCal}{\mathcal{S}}
\newcommand{\RCal}{\mathcal{R}}

\newcommand{\LL}{\mathcal{L}}
\newcommand{\pbm}{\mathbb{P}}
\newcommand{\T}{\mathbb{T}}
\newcommand{\E}{\mathbb{E}}
\newcommand{\Ef}{\mathcal{E}}

\newcommand{\R}{\mathbb{R}}

\newcommand{\ICal}{\mathcal I}

\newcommand{\CCal}{\mathcal C}

\newtheorem{rmk}{Remark}[section]

\newtheorem{thm}{Theorem}[section]

\newtheorem{prop}{Proposition}[section]
\newtheorem{defn}{Definition}[section]

\newtheorem{asp}{Assumption}[section]

\def\argmin_#1{\underset{#1}{\mathrm{arg\,min\, }}}
\def\argmax_#1{\underset{#1}{\mathrm{arg\,max\, }}}

\makeatletter
\def\dasharrowfill@#1#2#3#4{%
        $\m@th
        \thickmuskip0mu
        \medmuskip\thickmuskip
        \thinmuskip\thickmuskip
        \relax
        #4#1\mkern2mu
        \xleaders\hbox{$#4\mkern2mu#2\mkern2mu$}\hfill
        \mkern2mu
        #3$%
}

\def\dashleftarrowfill@{\dasharrowfill@\leftarrow\relbar\relbar}
\def\dashrightarrowfill@{\dasharrowfill@\relbar\relbar\rightarrow}
\def\dashleftrightarrowfill@{\dasharrowfill@\leftarrow\relbar\rightarrow}
\def\dashLeftarrowfill@{\dasharrowfill@\Leftarrow\Relbar\Relbar}
\def\dashRightarrowfill@{\dasharrowfill@\Relbar\Relbar\Rightarrow}
\def\dashLeftrightarrowfill@{\dasharrowfill@\Leftarrow\Relbar\Rightarrow}

\providecommand*\xdashleftarrow[2][]{%
  \ext@arrow 0055{\dashleftarrowfill@}{#1}{#2}}
\providecommand*\xdashrightarrow[2][]{%
  \ext@arrow 0055{\dashrightarrowfill@}{#1}{#2}}
\providecommand*\xdashleftrightarrow[2][]{%
  \ext@arrow 0055{\dashleftrightarrowfill@}{#1}{#2}}
\providecommand*\xdashLeftarrow[2][]{%
  \ext@arrow 0055{\dashLeftarrowfill@}{#1}{#2}}
\providecommand*\xdashRightarrow[2][]{%
  \ext@arrow 0055{\dashRightarrowfill@}{#1}{#2}}
\providecommand*\xdashLeftrightarrow[2][]{%
  \ext@arrow 0055{\dashLeftrightarrowfill@}{#1}{#2}}
\makeatother

\begin{document}

\title{Risk of Transfer Learning and its Applications in Finance}

\author{
Haoyang Cao
\thanks{Centre de Math\'ematiques Appliqu\'ees, Ecole Polytechnique.
\textbf{Email:}
haoyang.cao@polytechnique.edu}
\and
Haotian Gu
\thanks{Department of Mathematics, UC Berkeley.
\textbf{Email:} 
haotian{\textunderscore}gu@berkeley.edu }
\and
Xin Guo
\thanks{Department of Industrial Engineering \& Operations Research, UC Berkeley.
\textbf{Email:} 
xinguo@berkeley.edu}
\and
Mathieu Rosenbaum
\thanks{Centre de Math\'ematiques Appliqu\'ees, Ecole Polytechnique.
\textbf{Email:}
mathieu.rosenbaum@polytechnique.edu}
}
\date{\today
}
\maketitle

\begin{abstract}
Transfer learning is an emerging and popular paradigm for utilizing existing knowledge from  previous learning tasks to improve the performance of new ones.
In this paper,
we propose a novel concept of {\it transfer risk} and and analyze its properties to evaluate transferability of transfer learning. We apply transfer learning techniques and this concept of transfer risk to stock return prediction and portfolio optimization problems. Numerical results demonstrate a strong correlation between transfer risk and overall transfer learning performance, where transfer risk provides a computationally efficient way to identify appropriate source tasks in transfer learning, including cross-continent, cross-sector, and cross-frequency transfer for portfolio optimization.
\end{abstract}

\section{Introduction}
\paragraph{Transfer learning.} It is a popular paradigm in machine learning, with a simple idea:  leveraging knowledge from a well-studied learning problem (a.k.a. the source task) to enhance the performance of a new learning problem with similar features (i.e., the target task). In deep learning applications with limited and relevant data, transfer learning is a standard practice of utilizing large datasets (e.g., ImageNet) and their corresponding pre-trained models (e.g., ResNet50). It has enjoyed success across various fields, including natural language processing \citep{ruder2019transfer, devlin-etal-2019-bert}, sentiment analysis \citep{liu2019survey}, computer vision \citep{ganin2016domain, wang2018deep}, activity recognition \citep{cook2013transfer, wang2018stratified}, medical data analysis \citep{wang2022transfer, kim2022transfer}, bio-informatics \citep{hwang2010heterogeneous}, 
recommendation system \citep{pan2010transfer, yuan2019darec}, and fraud detection \citep{lebichot2020deep}.
See also various review papers such as \citep{survey1,tan2018survey,zhuang2020comprehensive} and the references therein. In the rapidly evolving AI landscape, where new machine learning techniques and tools emerge at a rapid pace, transfer learning is well suited as a versatile and enduring paradigm.  Meanwhile, the empirical successes of transfer learning has also encouraged theoretical studies of transfer learning, particularly in terms of quantifiable way of measuring whether transfer learning is suitable under given contexts; see for instance \citep{mousavi2020minimax}, \citep{nguyen2020leep}, \citep{you2021logme}, \citep{huang2022frustratingly}, \citep{nguyen2022generalization}, \citep{tripuraneni2020theory}, \citep{galanti2022generalization} and \citep{cao2023feasibility}.

\paragraph{Transfer learning in finance.}\label{subsec:tl-finance}
Transfer learning has recently gained its popularity in the field of finance, where limited data availability and excessive noise have hindered practitioners from accomplishing tasks such as equity fund recommendation \citep{zhang2018equity} and stock price prediction \citep{wu2022jointly,nguyen2019novel}.  Instead of starting from scratch for each specific task, it allows financial practitioners to capitalize on the knowledge and patterns accumulated from analogous tasks or domains, resulting in more accurate predictions and enhanced decision-making capabilities.

For instance,  \cite{zhang2018equity} addressed the issue of ``what to buy'' in equity fund investment by providing personalized recommendations; due the lack of transaction data in equity fund market, they utilized transfer learning and applied the profile of investors on the stock market to build that of the fund market; subsequently, this profile constituted an important role in the construction of the utility-based recommendation algorithm. \cite{leal2020learning} proposed a deep neural network controller for optimal trading on high frequency data; to overcome the scarcity of training data in high frequency trading, this deep neural network was first pretrained over simulated data, resulting in a good initialization for the fine-tuning process over genuine historical trading trajectories. In the work by \cite{wu2022jointly}, to improve the accuracy of stock trend prediction, the knowledge of industrial chain information was transferred to the prediction model via transfer learning: the deep learning models were first trained on stock indices of the upstream industry; then the best model, together with the model parameters, was transferred to predict the downstream industry, completing the industrial chain information transmission. \cite{doi:10.1080/14697688.2019.1622295} uncovered the existence of a universal price formation mechanism in financial markets via a large-scale deep learning model applied to a high-frequency database; they discovered that models trained over a dataset consisting of various types of assets exhibited superior generalization property compared with stock-specific models, and therefore provided empirical justification of the validity of applying transfer learning methods to financial problems. 

In fact, to assist in stock and market prediction, there have been a stream of works utilizing the advances in natural language processing to extract useful information from financial text. One such example is FinBERT \citep{liu2021finbert}, a financial text mining variant of the BERT model; to tackle the scarcity of labeled text data in the financial field, the authors designed six source tasks pretrained over large-scale general and domain-specific dataset, resulting in a financial text mining model outperforming the state-of-art models. For more examples and details, see survey papers on natural language based financial forecasting such as the work by \citep{xing2018natural}. 

Apart from predictive models, transfer learning helps improve trading decisions as well.  \cite{jeong2019improving} proposed a reinforcement learning-based trading system centering around a deep Q-network with a regressor network; with insufficient and highly volatile financial data, transfer learning techniques were adopted to overcome the overfitting problem. \cite{cartea2023bandits} proposed a two-layer data-driven execution algorithm:  the first strategic layer was to provide an optimal trading schedule for a sequence of orders; the second speculative layer was to employ a contextual bandit algorithm to output optimal execution strategy for each order; in order to allow transfer learning across different trading tasks in the second layer, the correlation among trading decisions sharing similar causal mechanisms was exploited for a better execution performance. 

Transfer learning techniques have also been applied to other areas of finance and economics:  \cite{lebichot2020deep} discussed and utilized domain adaptation, a special type of transfer learning technique, for the design of deep neural network-based automated fraud detection system so that companies were able to reuse the same pipeline to handle different payment systems; when utilizing quadratic rough Heston model to jointly calibrate SPX and VIX implied volatilities, transfer learning techniques were exploited to accelerate the training of the neural networks after adjusting the Hurst parameter in the model \citep{rosenbaum2021deep};  in order to improve the crude oil price prediction accuracy, \cite{cen2019crude} adopted the technique of  data transfer with prior knowledge to extend the size of training set for the deep learning prediction model consisting of long-short term memory units. 
  
\paragraph{Our work.}

 We propose a novel concept of {\it transfer risk}  to evaluate the potential benefit of transfer learning. Our form of transfer risk accounts for {\it both} the compatibility between the output and the input data {\it and} the compatibility between the models in the source and the target tasks,  allowing for the study of the trade-off between the two. We  establish several properties of transfer risk under the generic setting, including the continuity of transfer risk with respect to source tasks.  We also analyze the properties of two specific forms of transfer risk, namely the KL and Wasserstein-based transfer risk, and provide bound analysis for these two forms of transfer risk and  their relation. Moreover, we establish a connection between transfer risk and learning outcome for Gaussian-based models. These additional properties suggest  that our notion of transfer risk can be an computationally efficient indicator for the potential transfer learning performance and for  selecting proper source tasks for a given target task.

To test the relevance of this notion of transfer risk, we apply transfer learning techniques to two financial problems, namely, stock return prediction and portfolio optimization. 
For the stock prediction experiments, we first perform a signature transform over the daily return and volume data, formulating the prediction task as a regression problem; we then compare the prediction results between the direct learning  and the transfer learning approaches. The results show the consistency of the transfer risk with classical statistical metrics, and demonstrate that improved prediction accuracy can be achieved under appropriate transfer learning setting as opposed to the direct learning. 

For the portfolio optimization experiments, we test the performance of  transfer risk under three tasks, namely, cross-continent transfer,  cross-sector transfer,  and cross-frequency transfer. 
\begin{itemize}
    \item 
In the cross-continent transfer, which is to transfer a portfolio from the US equity market to other equity markets, our study shows different performances for different international markets. For instance,  transfer learning from the US market outperforms direct learning for Germany, but it performs relatively poorly for the Brazil market. This suggest that portfolios from the US market are better source tasks for the former than for the latter.
\item For the cross-sector transfer, which is to transfer a portfolio from one sector to different sectors, our analysis reveals that transfer risks in Health Care and Information Technology sectors display large negative correlations. In contrast, correlations are not significant for Utilities and Real Estate.

\item Regarding the cross-frequency transfer, which is to transfer a low-frequency portfolio to the mid-frequency domain, our results indicate that transferring a low-frequency portfolio (one-day) to higher frequencies (intraday) carries  high transfer risks with poor performances. In contrast, transferring between the mid and high-frequency regimes yields more robust and promising outcomes.
\end{itemize}

\section{Preliminary: Mathematical Framework of Transfer Learning}
 
 In a supervised setting, transfer learning consists of two tasks:  a source task $S$ and a target task $T$.
 The idea of transfer learning is to leverage the knowledge from the source task to improve the performance of the target task.
 
 If one fixes a probability space $(\Omega,\mathcal{F},\pbm)$, then the transfer learning can be formalized in an optimization framework as proposed by \cite{cao2023feasibility}. 
 In this optimization framework, the target task \(T\) is depicted as an optimization problem,
\begin{equation}\label{eq: obj-t}
    \min_{f\in A_T}\LL_T(f_T)=\min_{f_T\in A_T}\E[L_T(Y_T,f_T(X_T))].
\end{equation}
Here, $(X_T,Y_T)$ is a pair of $\XCal_T\times\YCal_T$-valued random variables, with $\XCal_T$ and $\YCal_T$ called target input and output spaces such that  $(\XCal_{T},\|\cdot\|_{\XCal_{T}})$ and $(\YCal_{T},\|\cdot\|_{\YCal_{T}})$ are Banach spaces with norms $\|\cdot\|_{\XCal_{T}}$ and $\|\cdot\|_{\YCal_{T}}$, respectively.
The function $L_T:\YCal_T\times\YCal_T\to\R$ is real-valued, and correspondingly $\LL_T(f_T)$ is a loss function measuring a model $f_T:\XCal_T\to\YCal_T$ for the target task $T$. The set $A_T$ denotes the collection of target models such that
 \begin{equation}\label{eq: a-t}
A_T\subset 
\{f_T|f_T:\XCal_T\to\YCal_T\}.
\end{equation}
Similarly, the source task \(S\) can be  defined  as an optimization problem of 
\begin{equation}\label{eq: obj-s}
    \min_{f_S\in A_S}\LL_S(f_S)= \min_{f\in A_S}\E[L_S(Y_S,f_S(X_S))].
\end{equation}
Here, $(X_S,Y_S)$ is a pair of $\XCal_S\times\YCal_S$-valued random variables, with $\XCal_S$ and $\YCal_S$ called source input and output spaces such that  $(\XCal_{S},\|\cdot\|_{\XCal_{S}})$ and $(\YCal_{S},\|\cdot\|_{\YCal_{S}})$ are Banach spaces with norms $\|\cdot\|_{\XCal_{S}}$ and $\|\cdot\|_{\YCal_{S}}$, respectively.
The function $L_S:\YCal_S\times\YCal_S\to\R$ is real-valued, and correspondingly $\LL_S(f_S)$ is a loss function measuring a model $f_S:\XCal_S\to\YCal_S$ for the source task $S$. The set $A_S$ denotes the collection of source models such that
 \begin{equation}\label{eq: a-s}
A_S\subset 
\{f_S|f_S:\XCal_S\to\YCal_S\}.
\end{equation}

The mathematical framework of transfer learning  \citep{cao2023feasibility} is summarized as follows.
\begin{equation}\label{eq: tl-fw}
    \begin{matrix}
        \XCal_S\ni X_S & \xRightarrow{\text{\hspace{4pt} Pretrained model } f_S^* \text{ from } \eqref{eq: obj-s}\text{\hspace{4pt}}} & f_S^*(X_S)\in\YCal_S\\
       T^X\Big\Uparrow & & \Big\Downarrow T^Y \\
        \XCal_T\ni X_T & \xdashrightarrow[\text{\hspace{10pt}}f_T^*\in\argmin_{f\in A_T} \LL_T(f_T)\text{\hspace{10pt}}]{\text{Direct learning \eqref{eq: obj-t} }} & f_T^*(X_T)\in\YCal_T
    \end{matrix}
\end{equation}
Here, the ``Direct learning''  is to directly analyze  \eqref{eq: obj-t} and solve for an optimizer $f_T^*$, with $\pbm_T=Law(f_T^*(X_T))$ as the probability distribution of its output. 
Along the alternative route of transfer learning, the optimal source model \(f_S^*\) is also referred to as a pretrained model, and we use $\pbm_S=Law(f_S^*(X_S))$ to denote the probability distribution of its output.

Moreover, \(T^X:\XCal_T\to\XCal_S\) is called 
the input transport mapping and \(T^Y:\XCal_T\times\YCal_S\to\YCal_T\) is the output transport mapping, where $\T^X$ and $\T^Y$ are proper sets of transport mappings such that 
\[T^X\subset\{f_\text{input}|f_\text{input}:\XCal_T\to\XCal_S\},\]
and
\[\left\{T^Y(\cdot, (f_S^*\circ T^X)(\cdot))|T^X\in\T^X,T^Y\in\T^Y\right\}\subset A_T.\]
In particular, when $\XCal_S=\XCal_T$ (resp. $\YCal_S=\YCal_T$), the identity mapping  $id^X(x)=x$ (resp.  $id^Y(x,y)=y$) is included in $\T^X$ (resp. $\T^Y$).
Within this mathematical framework, The transfer learning  \eqref{eq: tl-fw} with supervised setting can be  defined as follows.
\begin{defn}[Transfer learning with supervised setting]\label{def:tl}
The transfer learning procedure presented in \eqref{eq: tl-fw} is to solve the optimization problem
\begin{align}
\label{eq: doub-trans}
\min_{T^X\in\mathbb{T}^X,T^Y\in\mathbb{T}^Y}\LL_T\left(T^Y(\cdot, (f_S^*\circ T^X)(\cdot))\right)=\E\left[L_T\left(Y_T,T^Y(X_T,(f_S^*\circ T^X)(X_T))\right)\right].
\end{align}
\end{defn}
This transfer learning framework \eqref{eq: doub-trans} in Definition \ref{def:tl}  can be easily extended to an unsupervised setting, where \(X_{\cdot}\in\XCal_{\cdot}\) is available but \(Y_{\cdot}\in\YCal_{\cdot}\) is not.

\begin{defn}[Transfer learning with unsupervised setting]
     \label{rmk:tl-unsup}
     Transfer learning with unsupervised setting is defined as 
     \begin{equation}
         \label{eq: doub-trans-unsup}
         \min_{T^X\in\T^X,T^Y\in\T^Y}\LL_T(T^Y(\cdot,(f_S^*\circ T^X)(\cdot))),
     \end{equation}
     with the loss function \(\LL_{\cdot}:A_{\cdot}\to\R\) depending only on \(Law(X_{\cdot})\).
\end{defn}
In Sections \ref{sec:return_predict} and \ref{sec:finance_PO}, we will present two concrete examples of the transfer learning framework (\ref{eq: tl-fw}), one under supervised setting and the other under unsupervised setting. 

The procedure of solving the above optimization problem (either \eqref{eq: doub-trans} or \eqref{eq: doub-trans-unsup}) is often referred to as fine-tuning in the literature of transfer learning. It is to choose some initial transport mappings $T^X_0\in\T^X_0\subset\T^X$ and $T^Y_0\in\T^Y_0\subset\T^Y$ to derive an intermediate model $f_{ST}\in A_T$ with
\begin{equation}\label{eq: int-model}
    f_{ST}(x)=T^Y_0(x, (f_S^*\circ T^X_0)(x)),\quad \forall x\in\XCal_T,
\end{equation}
with the set of possible intermediate models denoted as
\begin{equation}\label{eq: int-set}
    \ICal=\left\{T^Y_0(\cdot, (f_S^*\circ T^X_0)(\cdot))\big|T^X_0\in\T^X_0, T^Y_0\in\T^Y_0\right\}.
\end{equation}
This fine-tuning procedure allows for computational efficiency in terms of {\it transfer risk}, introduced  in the next  section. 

\section{Transfer Risk}
\label{sec:transferrisk}
In parallel to transfer learning framework \eqref{eq: tl-fw}, there are two major sources of transfer risk for a fixed intermediate model $f_{ST}\in\ICal$:  the risk that measures the mismatch  between the output distributions of the intermediate model $f_{ST}$ and the optimal target model $f_T^*$, and the risk reflecting the difference between the transported target input and the source input.

\begin{defn}[Output transport risk]\label{defn:outputrisk}
Let \(\Ef^O:A_T\to\R\) be a real-valued function on the set of target models. For any \(f_{ST}\in\ICal\subset A_T\), $\Ef^O(f_{ST})$ is called an output transport risk of intermediate model $f_{ST}$ if it  satisfies
    \begin{enumerate}
        \item $\Ef^O(f_{ST})\geq 0$, i.e., transfer learning always incurs a non-negative effort;
        \item $\Ef^O(f_{ST})=0$ if and only if $\pbm_T=\pbm_{ST}$, where $\pbm_T:=Law(f_{T}(X_T))$ and $\pbm_{ST}:=Law(f_{ST}(X_T))$. That is, the output transport risk vanishes when the intermediate model $f_{ST}$ completely recovers the distribution of the optimal target task.
    \end{enumerate}
\end{defn}

Clearly, the lower this output risk, the more effective the transfer scheme with the intermediate model $f_{ST}$.

\begin{defn}[Input transfer risk]\label{defn: inputrisk}
Let \(\Ef^I:\T^X\to\R\) be a real-valued function on the set of input transport mappings. Given an import transport mapping $T^X_0\in\T^X_0\subset \T^X$, $\Ef^I(T^X_0)$ is called an input transport risk if it  satisfies
    \begin{enumerate}
        \item $\Ef^I(T^X_0)\geq 0$, i.e., transfer learning always incurs a non-negative effort;
        \item $\Ef^I(T^X_0)=0$ if and only if $T^X_0\#Law(X_T)=Law(X_S)$. 
    \end{enumerate}
\end{defn}
Evidently, the lower this input risk, the higher the similarity between the transported target input $T^X_0(X_T)$ and the source input $X_S$.



Both the input transfer risk and the output transfer risk  characterize the divergence between probability distributions, and their exact forms can be task dependent.  Nevertheless,  there is a key difference between these two forms of risks:  in the output transport risk, $\pbm_T$, the output distribution of the optimal target model, is {\it unknown}, and no prior knowledge about $f_T^*$ is assumed. 
Therefore, analyzing the output transport risk is decisively more complicated, as will be clear  in Section \ref{sec:divergence}.

We are now ready to propose the notion of  {\it transfer risk}  by considering all intermediate models in $\ICal$, in order to  measure the effectiveness of a transfer learning \eqref{eq: doub-trans} (depicted in Figure \ref{fig:risk_motivation}).
\begin{defn}[Transfer risk]\label{defn: trans-bene}
For a transfer learning procedure characterized by the 6-tuple ${(S,T,\T^X,\T^X_0,\T^Y,\T^Y_0)}$  in \eqref{eq: doub-trans}, the transfer risk of the transfer learning framework \eqref{eq: doub-trans} from source task $S$ to target task $T$ is defined as $\CCal:\R\times\R\to\R$ with $\CCal(0,0)=0$ such that  
\begin{equation}\label{eqn:risk}
\CCal(S,T)=\inf_{f_{ST}\in\ICal}\CCal(S,T|f_{ST}).
\end{equation}
Here, for a given $f_{ST}=T^Y_0(\cdot, (f_S^*\circ T^X_0)(\cdot))\in\ICal$, $\CCal(S,T|f_{ST})\ge 0 $ is called  model-specific transfer risk with the following properties:
\begin{enumerate}
    \item $\CCal(S,T|f_{ST})=C(\Ef^O(f_{ST}),\Ef^I(T^X_0))$ is non-decreasing in $\Ef^O(f_{ST})$ under any fixed $\Ef^I(T^X_0)$ and non-decreasing in $\Ef^I(T^X_0)$ under any fixed $\Ef(f_{ST})$; 
    \item $\CCal(S,T|f_{ST})$ is Lipschitz  in the sense that for any other transfer problem characterized by $(\bar S,\bar T,\bar\T^X,\bar\T_0^X,\bar\T^Y,\bar\T_0^Y)$ and one of its intermediate models $\bar f_{ST}=\bar T^Y_0(\cdot, (\bar f_S^*\circ \bar T^X_0)(\cdot))\in\bar\ICal$, there exists a constant $L>0$ such that 
    \[\begin{aligned}|\CCal(S,T|f_{ST})-\CCal(\bar S,\bar T|\bar f_{ST})|&\leq L(|\Ef^O(f_{ST})-\Ef^O(\bar f_{ST})|\\
    &+|\Ef^I(T^X_0)-\Ef^I(\bar T^X_0)|).\end{aligned}\]
\end{enumerate}
\end{defn}

Note that the Lipschitz condition in Definition \ref{defn: trans-bene} is to emphasize the dependence of transfer risk on a given transfer learning problem. This Lipschitz property is satisfied when the function $C$ in Definition \ref{defn: trans-bene} is Lipschitz continuous. 

Note also these definitions of risks involve the sets of initial transport mappings $\T^X_0$ and $\T_0^Y$, instead of the sets of all possible transport mappings $\T^X$ and $\T^Y$. These reduced sets allow for efficient evaluation of transfer risk prior to starting the full-scale transfer learning.

\begin{figure}[!ht]
    \centering
    \includegraphics[width=0.6\textwidth]{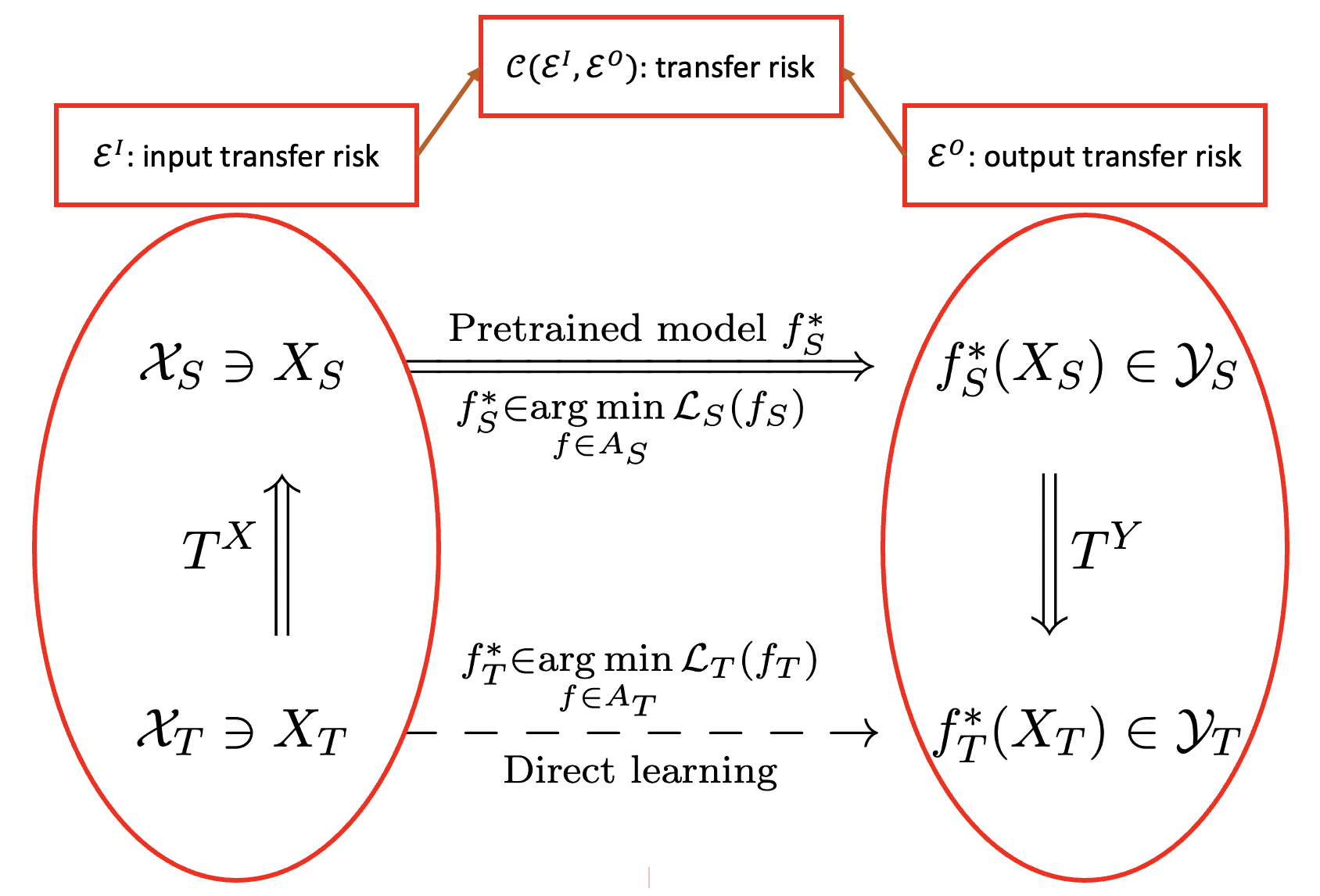}
    \caption{Parallel relation between transfer risk and transfer learning framework}
    \label{fig:risk_motivation}
\end{figure}


One simple example of the model-specific transfer risk is
\begin{equation}\label{eq: risk-ln}
    \CCal^\lambda(S,T|f_{ST})=\Ef^O(f_{ST})+\lambda \Ef^I(T^X_0),
\end{equation}
where $\lambda>0$ is a pre-specified parameter modulating the weight of the input transport in the transfer learning problem \eqref{eq: tl-fw}. 

\subsection{Properties of transfer risk}
We will show that  transfer risk is continuous in the input distribution and robust with respect to the pretrained model. 
This property is useful to measures how transfer risk is affected by the choice of source task $S$ for a fixed target task. It can also be used to exclude {\it a priori} inappropriate source tasks when compared against existing viable source tasks.

To study the continuity of the transfer risk, one needs to assign an appropriate metric for $S$. To this end, recall that  $S$ is determined by $(Law(X_S),f_S^*)$:
the probability distribution of source input $Law(X_S)$ and  the pretrained model $f_S^*$ in \eqref{eq: obj-s}. That is, for a 
target task $T$, the risk $\CCal(S,T)$ can be rewritten as
$\CCal(S,T)=\CCal(S)=\CCal(\mu,f)$ for any $S=(\mu,f)\in\SCal$, with 
$\SCal\subset\mathcal{P}(\XCal_S)\times A_S$, where \(A_S\) is the set of source models given in \eqref{eq: a-s}.

More specifically, for any $S_1,S_2\in\SCal$ such that $S_1=(\mu_1,f_1)$ and $S_2=(\mu_2,f_2)$, define the metric $d_S$ over $\SCal$ as
\begin{equation}\label{eqn:d_s}
d_S(S_1,S_2):=D(\mu_1,\mu_2)+d_M(f_1,f_2),\quad \forall\mu_1,\mu_2\in\mathcal{P}(\XCal),\ \ \forall f_1,f_2\in A_S.
\end{equation}
Here  $D:\mathcal{P}(\XCal_S)\times \mathcal{P}(\XCal_S)\to\R$ is a metric function for $\mathcal{P}(\XCal_S)$, the set of probability measures over $\XCal_S$; and the metric function \(d_M\) for \(A_S\) is defined as
\[d_M(f_1,f_2):=\min\{M,\sup_{x\in\XCal_S}\|f_1(x)-f_2(x)\|_{\YCal_S}\},\quad \forall f_1,f_2\in A_S\]
for a sufficiently large constant $M>0$.

In the following discussion on continuity, the next assumption is necessary. Assumption \ref{ass:D} ensures that the choice of input transfer risk is consistent with the metric $d_S$ in \eqref{eqn:d_s} defined between source tasks.

\begin{asp}\label{ass:D}
    For any input transport mapping $T^X_0\in\mathbb{T}^X_0$, assume the input transfer risk $\Ef^I(T^X_0)$ take the form  $\Ef^I(T^X_0):=D(T^X_0\#Law(X_T),Law(X_S))$, where $D:\mathcal{P}(\XCal_S)\times \mathcal{P}(\XCal_S)\to\R$ is the distance function appearing in \eqref{eqn:d_s}.
\end{asp}

By definition, the following degenerate case holds immediately.
\begin{prop}[Zero transfer risk]
    Suppose $\XCal_T=\XCal_S$, $\YCal_T=\YCal_S$ and the target task $T\in\SCal$. Then $\CCal(T)=0$.
\end{prop}

Now, we consider source tasks $S_1,S_2\in\SCal$ that differ only in the input distribution, i.e., $S_f^1=(\mu_1,f)$ and $S_f^2=(\mu_2,f)$. Then we have the following continuity property for $\CCal$.
\begin{prop}[Continuity in input distribution]\label{lem: cont-distr}
    Assume Assumption \ref{ass:D}. Fix $f\in A_S$. $\CCal(\cdot, f)$ is continuous on $(\mathcal{P}(\XCal_S),D)$.
\end{prop}

\begin{proof}[Proof of Proposition \ref{lem: cont-distr}]
Fix an arbitrary $\epsilon>0$. Take any $\mu\in\mathcal{P}(\XCal_S)$.
We first establish the lower semi-continuity: For any $T^X_0\in\T_0^X$ and $T_0^Y\in\T_0^Y$, let $f_{I}$ denote the corresponding intermediate model from source model $f$. By Definition \ref{defn: trans-bene}, we have
\[\CCal(\mu,f)-\frac{1}{2}\epsilon<\CCal(\mu,f|f_{I}).\]
By the triangle inequality of $D$ and the Lipschitz property of $\CCal(\mu,f|f_I)$, take $\delta=\frac{\epsilon}{2L}$ for any $\mu'\in B_\delta(\mu)\subset \mathcal{P}(\XCal_S)$,
\[\CCal(\mu,f|f_{I})\leq \CCal(\mu',f|f_{I})+L\delta.\]
Notice that the choice of $\delta$ is independent of $T_0^X$ and $T_0^Y$. Therefore, 
\[\CCal(\mu,f)-\epsilon<\CCal(\mu',f;f_{I})\Rightarrow \CCal(\mu,f)-\epsilon<\CCal(\mu',f).\]

Now we show the upper semi-continuity.  By the definition  of $\CCal$, there exists $\bar T_0^X\in\T_0^X$ and $\bar T_0^Y\in\T_0^Y$, with corresponding intermediate model $\bar f_I$, such that 
\[\CCal(\mu,f|\bar f_{I})<\CCal(\mu,f)+\frac{1}{2}\epsilon.\]
Again, by the triangle inequality of $D$ and the Lipschitz property of $\CCal(\mu,f|f_I)$, take $\delta=\frac{\epsilon}{2L}$ for any $\mu'\in B_\delta(\mu)\subset \mathcal{P}(\XCal_S)$,
\[\CCal(\mu,f|\bar f_{I})\geq \CCal(\mu',f|\bar f_{I})-\delta.\]
Then we have 
\[\CCal(\mu',f)\leq \CCal(\mu',f|\bar f_{I})<\CCal(\mu,f)+\epsilon.\]

\end{proof}

This proposition shows that transfer risk will change continuously along with any modification in source input. The sensitivity of transfer risk with respect to the change in source input distribution depends on the Lipschitz constant $L$ of $\CCal$. Therefore, one can modulate this sensitivity by carefully designing the $C$ function in Definition \ref{defn: trans-bene}. For instance, for linear transfer risk $\CCal^\lambda$ in \eqref{eq: risk-ln}, the sensitivity can be controlled by varying the value of $\lambda$.

Next, consider source tasks $S_1,S_2\in\SCal$ that differ only in the pretrained model, i.e., $S_\mu^1=(\mu,f_1)$ and $S_\mu^2=(\mu,f_2)$. Then we have  the robustness of the transferability in terms of  the continuity of $\CCal(\mu,\cdot)$ in pretrained model $f\in(A_S,d_M)$. 

\begin{prop}[Continuity in pretrained model]\label{lem: cont-f} 
Assume Assumption \ref{ass:D}, and assume that there exists a constant $L>0$ such that for any $T_0^Y\in\T_0^Y$, 
\[T_0^Y(x_1,y_1)-T_0^Y(x_2,y_2)\leq L\left(\|x_1-x_2\|_{\XCal_T}+\|y_1-y_2\|_{\YCal_S}\right),\]
for all $(x_1,y_1),(x_2,y_2)\in\XCal_T\times\YCal_S$.
Assume also that there exist  some $L'>0$ and $p\geq1$ such that the output transfer risk satisfies
    \[\left|\Ef^O(h_1)-\Ef^O(h_2)\right|\leq L'\mathcal{W}_p(h_1\#Law(X_T),h_2\#Law(X_T))^p\]
    for all $h_1,h_2\in\mathcal I$. 
    Then $\CCal(\mu,\cdot)$ is continuous on $(A_S,d_M)$ for any fixed $\mu\in\mathcal{P}(\XCal_S)$.
\end{prop}

\begin{proof}[Proof of Proposition \ref{lem: cont-f}]
Take any $T_0^X\in\T_0^X$ and $T_0^Y\in\T_0^Y$. For any $f_1,f_2\in(A_S,d_M)$, denote their corresponding intermediate model as $f_I^1$ and $f_I^2$, respectively. Then
\[\begin{aligned}
|\Ef^O(f_I^1)-\Ef^O(f_I^2)|&\leq L'W_p(f_I^1\#Law(X_T),f_I^2\#Law(X_T))^p\\
&=L'\inf_{\pi\in\Pi(f_I^1\#Law(X_T),f_I^2\#Law(X_T))}\int_{\YCal_T\times\YCal_T}\|x-y\|_{\YCal_T}^p\pi(dx,dy)\\
&\leq L'\inf_{\gamma\in\Pi(Law(X_T),Law(X_T))}\int_{\XCal_T\times\XCal_T}\|T_0^Y(x,f_1(T_0^X(x)))-T_0^Y(y,f_2(T_0^X(y)))\|_{\YCal_T}^2\pi(dx,dy)\\
&\leq 2^{p-1}L^pL'\left[\inf_{\gamma\in\Pi(Law(X_T),Law(X_T))}\int_{\XCal_T\times\XCal_T}\|x-y\|_{\XCal_T}^pd\pi(dx,dy)+d_M(f_1,f_2)^p\right]\\
&=2^{p-1}L^pL'\left[W_p(Law(X_T),Law(X_T))^p+d_M(f_1,f_2)^p\right]=2^{p-1}L^pd_M(f_1,f_2)^p.
\end{aligned}\]
The rest of the proof is similar to that of Proposition \ref{lem: cont-distr}.
\end{proof}

This proposition shows that transfer risk will change continuously along with the modification in the pretrained model. The sensitivity of transfer risk with respect to the change in pretrained model is determined by three factors: (1) the Lipschitz constant inherited from the $C$ function in Definition \ref{defn: trans-bene}, (2) the choice of output transport risk $\Ef^O$, and (3) the family of output transport mappings $\T_0^Y$. In practice, one may control the sensitivity of the transfer risk through careful choices of those quantities.


Propositions \ref{lem: cont-distr} and \ref{lem: cont-f} lead to the following results.
\begin{prop}
    \label{thm: cont}
    Given the conditions in Proposition \ref{lem: cont-f}. Then the transfer risk $\CCal$ as in Definition \ref{defn: trans-bene} is continuous on $(\SCal, d_S)$.
\end{prop}

Propositions \ref{lem: cont-distr} -- \ref{thm: cont} reveals that under a given target task, transfer risk is continuously influenced by both the changes in the source input and the pretrained model. Therefore, transfer risk is to evaluate the suitability of performing transfer learning and the  appropriate choice of given source tasks for a  target task.

\subsection{Transfer Risk under KL-Divergence and Wasserstein Distance}
\label{sec:divergence}
It is clear that  different learning tasks may require  different choices of divergence functions for assessment of transfer risk. In this section, we will focus on transfer risk under two types of divergence functions, namely, KL-divergence and Wasserstein distance. We will present their respective properties and their relation.

\paragraph{KL-based output transport risk.} For learning tasks such as the classification problem, one may use  cross-entropy as the loss function. 

Specifically, let $\pbm_T=\tilde\pbm_T+\pbm_0$ be its unique Lebesgue decomposition, i.e., for any measurable set $B\subset\YCal_T$, there exists some  function $h_{ST}:\YCal_T\to\R^+$ such that
$\tilde\pbm_T(B)=\int_Bh_{ST}d\pbm_{ST}$, with $\pbm_0$ singular with respect to $\pbm_{ST}$. Then the KL-based output risk can be  defined as
    \[
    \Ef^O_{KL}(f_{ST}):=D_{KL}(\tilde\pbm_T\|\pbm_{ST})+{H}(\pbm_0),
    \]
where $H(\pbm_0)$ is the entropy function of $\pbm_0$. 

\begin{prop}\label{lem: kl-heur}
    For a classification problem over $K\in\mathbb{N}$ classes with cross entropy as the training loss,  for any $f_{ST}\in\mathcal{I}$, 
    \begin{eqnarray*}
    \sum_{i=1}^K\log p_{ST}(i) \leq H(\pbm_T,\pbm_{ST})-H(Law(Y_T),\pbm_{ST}) \leq -\sum_{i=1}^K\log p_{ST}(i),
    \end{eqnarray*}
    where $p_{ST}$ denotes the probability mass function for $\pbm_{ST}$.
\end{prop}
Note that $H(Law(Y_T),\pbm_{ST})$ is indeed the cross-entropy loss for the classifier $f_{ST}$. Therefore, in actual training, one may use $H(Law(Y_T),\pbm_{ST})\pm\sum_{i=1}^K\log p_{ST}(i)$ to replace   $\Ef^O_{KL}(f_{ST})$.

It is well-known that KL divergence belongs to the class of $f$ divergence and therefore it inherits properties of $f$  divergence. 
To emphasize the impact of source task over transfer risk under a given target task $T$, we denote the KL-based transfer risk as follows
\[\CCal_{KL}(S|T_0^X,T_0^Y)=\CCal_{KL}(S,T|f_I)=C\left(\Ef^I_{KL}(T_0^X),\Ef^O_{KL}(f_I)\right)\]
for any $S=(\mu,f)\in\SCal$, $T_0^X\in\T^X_0$,  $T_0^Y\in\T_0^Y$, $f_I(\cdot)=T_0^Y(\cdot,f\circ T_0^X(\cdot))\in\ICal$, and 
\[\CCal_{KL}(S)=\CCal_{KL}(S,T)=\inf_{f_I\in\ICal}\CCal_{KL}(S,T|f_I).\]
\begin{prop}[Convexity]
    Fix $T_0^X\in\T_0^X$ and $T_0^Y\in\T_0^Y$, and assume the following conditions:
    \begin{enumerate}
        \item The container function $C$ in definition \ref{defn: trans-bene} is jointly convex;
        \item $T^Y_0$ is linear in the second argument, i.e., for any constants $k_1,k_2\in\R$,
        \[T^Y_0(x,k_1y_1+k_2y_2)=k_1T_0^Y(x,y_1)+k_2T^Y_0(x,y_2),\quad \forall x\in\XCal_T\,y1,y_2\in\YCal_S.\]
    \end{enumerate}
    Then transfer risk $\CCal_{KL}(\cdot|T_0^X,T_0^Y)$ is a convex function in $S\in\SCal$.
\end{prop}
\begin{proof}
    Fix any $S_i=(\mu_i,f_i)\in\SCal$, $i=1,2$, and $\alpha\in[0,1]$ such that $S_\alpha=(\mu_\alpha,f_\alpha)\in\SCal$, where $\mu_\alpha=\alpha\mu_1+(1-\alpha)\mu_2$ and $f_\alpha=\alpha f_1+(1-\alpha)f_2$. The corresponding intermediate model for $S_u$ is given by $f_{I,u}(\cdot)=T_0^Y(\cdot,f_u\circ T_0^X(\cdot))\in\ICal$, $u=1,2,\alpha$. We have 
    \[\begin{aligned}
        \CCal_{KL}(S_\alpha|f_{I,\alpha})&=C\left(\Ef^{I,\alpha}_{KL}(T_0^X),\Ef^O_{KL}(f_{I,\alpha})\right)\\
        &\overset{(a)}{\leq}C\left(\alpha\Ef^{I,1}(T_0^X)+(1-\alpha)\Ef^{I,2}(T_0^X),\alpha\Ef^{O}(f_{I,1})+(1-\alpha)\Ef^O(f_{2})\right)\\
        &\overset{(b)}{\leq}\alpha\CCal_{KL}(S_1|f_{I,1})+(1-\alpha)\CCal_{KL}(S_2|f_{I,2}),
    \end{aligned}\]
    where $\Ef^{I,u}(T_0^X)=D_{KL}(Law(T_0^X(X_T))\|\mu_u)$ for $u=1,2,\alpha$. Here, (a) is due the joint convexity of $f$ divergence, the linearity of $T_0^Y$, and the monotonicity of $C$, and (b) is due the joint convexity of $C$.
\end{proof}

\paragraph{Wasserstein-based output transport risk.}
For learning problems such as GANs or supervised learning with domain adaption, Wasserstein and related distances are popular choices to measure the distance between the generative distribution and the target distribution. Therefore, a Wassertein-based output risk is a natural choice  for  such learning targets. 

More specifically, for $p\geq1$, let $\mathcal{P}_p(\YCal_T)$ be the set of probability measures over $\YCal_T$ such that
\[\int_{\R^{d_{O,T}}}\|x\|_{\YCal_T}^pd\mu(x)<\infty,\quad\forall\mu\in\mathcal{P}_p(\YCal_T).\]
The Wasserstein-based output risk is defined as
\begin{align}\label{eqn:W_risk}
\Ef^O_{W}(f_{ST}):=\mathcal{W}_p(\pbm_{ST},\pbm_{T})^p:=\inf_{\gamma\in\Pi(\pbm_{ST},\pbm_{T})}\int_{\R^{d_{O,T}}\times\R^{d_{O,T}}}\|x-y\|_{\YCal_T}^pd\gamma(dx,dy),
\end{align}
for some suitable choice of $p\geq1$, where $\Pi(\pbm_{ST},\pbm_{T})$ denotes the set of couplings of probability measures $\pbm_{ST}$ and $\pbm_{T}$.

Analogy to Proposition \ref{lem: kl-heur} is the following property for $\Ef^O_W(f_{ST})$, based on the triangle inequality of the Wasserstein distance.
\begin{prop}\label{lem: w-heur}
    The Wasserstein-based output risk $\Ef^O_W$ in (\ref{eqn:W_risk}) is upper bounded in the following sense:
    \[\Ef^O_W(f_{ST})\leq 2^{p-1}[\mathcal{W}_p(\pbm_{ST},Law(Y_T))^p+\mathcal{W}_p(\pbm_{T},Law(Y_T))^p].\]
\end{prop}
 
Now, consider any intermediate model $f_{ST}$, then Talagrand's inequality 
\citep{talagrand1996transportation} gives 
$$\Ef^I_W(T_0^X)\leq 2\Ef^I_{KL}(T_0^Y),\Ef^O_W(f_{ST})\leq 2\Ef^O_{KL}(f_{ST}).$$
In particular, the linear transfer risk defined in \eqref{eq: risk-ln} satisfies
\begin{align}\label{eqn:risk-ln-KLW}
    \CCal^\lambda_{W}(S,T|f_{ST}):=\Ef^O_W(f_{ST})+\lambda\cdot\Ef^I_W(T_0^X)\leq 2\CCal^\lambda_{KL}(S,T|f_{ST}):=2(\Ef^O_{KL}(f_{ST})+\lambda\cdot\Ef^I_{KL}(T_0^X)).
\end{align}
Such a relation between KL- and Wasserstein-based linear transfer risks \eqref{eqn:risk-ln-KLW} gives the following proposition.
\begin{prop}\label{prop: talagrand}
Consider transfer risk in linear form as in \eqref{eqn:risk-ln-KLW}. Suppose $\YCal_T$ is a finite-dimensional Euclidean space and $\pbm_T\ll\pbm_{ST}$. Then for a given transfer learning problem $(S,T,\T_X,\T_Y,\T_X^0,\T_Y^0)$, 
\[\CCal_W(S,T)\leq 2\CCal_{KL}(S,T).\]
\end{prop}

\subsection{Transfer Risk and Regret in Gaussian-based Models}\label{subsec:gaussian}
We further analyze the properties of transfer risk by  establishing a connection between  transfer risk \eqref{eqn:risk} and  transfer learning performance in Gaussian-based models. Due to Proposition \ref{prop: talagrand}, we focus on Wasserstein-based transfer risk.

 Consider a source task $S$ and a target task $T$ with the input space $\XCal_S=\XCal_T=\R^d$ and the output space $\YCal_S=\YCal_T=\R$. We assume that both source and target data satisfy  $(d+1)$-dimensional Gaussian distributions such that $(X_\cdot,Y_\cdot)\sim \mathcal{N}(\mu_\cdot,\Sigma_\cdot)$, with
\begin{equation}\label{eq: source-data}
\mu_\cdot=\begin{pmatrix}\mu_{\cdot,X}\\\mu_{\cdot,Y}\end{pmatrix},\quad \Sigma_\cdot=\begin{pmatrix}\Sigma_{\cdot,X}&\Sigma_{\cdot,XY}\\\Sigma_{\cdot,YX}&\Sigma_{\cdot,Y}\end{pmatrix},
\end{equation}
where
$\mu_{\cdot,Y}\text{ and }\Sigma_{\cdot,Y}\in\R$, $\mu_{\cdot,X}\text{ and }\Sigma_{\cdot,XY}\in\R^d$, $\Sigma_{\cdot,YX}=\Sigma_{\cdot,XY}^\top$, and $\Sigma_{\cdot,X}\in\R^{d\times d}$. Assume that the sets of admissible source and target models $A_S=A_T=\{f:\R^d\to\R\}$. For any $f\in A_S=A_T$, assume  that the loss functions are
\begin{equation}\label{eq: lrloss}
    \LL_S(f)=\E\|Y_S-f(X_S)\|_2^2, \ \ \LL_T(f)=\E\|Y_T-f(X_T)\|_2^2.
\end{equation}
Now, consider a simple setting where the input (resp. output) transport set $\mathbb{T}^X$ (resp. $\mathbb{T}^Y$) is a singleton set only containing the identical mapping on $\R^d$ (resp. $\R$). Then, the transfer learning scheme \eqref{eq: doub-trans} is equivalent to directly applying the optimal source model $f^*_S$ to the target task. Consequently, the intermediate model set $\mathcal{I}$ in \eqref{eq: int-set} is also a singleton set with $\mathcal{I}=\{f^*_S\}$. 

Define the transfer risk in this  problem as the Wasserstein-based output transport risk:
\begin{equation}\label{eqn:linear_risk}
\CCal_W(S,T)=\CCal_W(S,T|f^*_S)=\Ef^O_W(f^*_S).
\end{equation}
Meanwhile,  
define the \textit{regret} as the gap between the transfer learning  and the direct learning:
\begin{equation}\label{eqn:regret}
    \RCal(S,T):=\LL_T(f^*_S)-\LL_T(f^*_T).
\end{equation}
Then, the following theorem shows that the transfer risk serves as a lower bound of the regret.


\begin{thm}\label{prop: lb}
    For transfer learning in the Gaussian model \eqref{eq: source-data}-\eqref{eq: lrloss}, the regret with respect to the chosen intermediate model $\RCal(S,T)$ in \eqref{eqn:regret} is lower bounded by the Wasserstein-based transfer risk in \eqref{eqn:linear_risk},
    \[\CCal_W(S,T)\leq \RCal(S,T).\]
\end{thm}

Theorem \ref{prop: lb} suggests that transfer risk provides a preliminary indication of the effectiveness of transfer learning, especially for eliminating unsuitable candidate pretrained models or source tasks when the transfer risk is high. 

\begin{proof} It is without  loss of generality that we assume source and target data are of matching dimensions, as otherwise analysis can be modified straightforwardly.
Then direct computation shows that the optimal source model 
\begin{equation}\label{eq: lr-opt-prob}
f_S^*\in\argmin_{f\in A_S}\LL_S(f)
\end{equation}
is given by 
\begin{equation}\label{eq: source-opt-model}
    f_S^*(x)=w_S^\top x+b_s,
\end{equation}
 where
 \begin{equation}\label{eq: pretrn-model}
        w_S=\Sigma_{S,X}^{-1}\Sigma_{S,XY}\in\R^d,\quad b_S=\mu_{S,Y}-\Sigma_{S,YX}\Sigma_{S,X}^{-1}\mu_{S,X}\in\R.
    \end{equation}
    

Meanwhile, the optimal target model $f_T^*$ is given by
\begin{equation}\label{eq: opt-target-model}
    f_T^*(x)=w_T^\top x+b_T,\quad \forall x\in\R^d,
\end{equation}
where
\begin{equation}\label{eq: target-opt-param}
    w_T=\Sigma_{T,X}^{-1}\Sigma_{T,XY},\quad b_T=\mu_{T,Y}-\Sigma_{T,YX}\Sigma_{T,X}^{-1}\mu_{T,X}.
\end{equation}
The corresponding output distribution is then given by
\begin{equation}\label{eq: target-pred}
\pbm_T=\E[Y|X]=N(w_T^\top\mu_{T,X}+b_T,w_T^\top\Sigma_{T,X}w_T)=N(\mu_T,w_T^\top\Sigma_{T,X}w_T).
\end{equation}


Given the optimal models in both the source task and the target task, specified by \eqref{eq: source-opt-model}-\eqref{eq: pretrn-model} and \eqref{eq: opt-target-model}-\eqref{eq: target-opt-param}, we have
\begin{equation}\label{eq: tl-init-distrn}
        \pbm_{ST}=f_S^*\# N(\mu_{T,X},\Sigma_{T,X})=N(w_S^\top\mu_{T,X}+b_S, w_S^\top\Sigma_{T,X}w_S).
\end{equation}
Notice that $\pbm_T\ll\pbm_{ST}$, therefore the Lebesgue decomposition leads to 
    $\pbm_T=\tilde\pbm_T$, such that
\begin{equation}
    \label{eq: leb-decomp}
    \frac{d\tilde\pbm_T(y)}{d\pbm_{ST}(y)}=h_{ST}(y)=\sqrt\frac{w_S^\top\Sigma_{T,X}w_S}{w_T^\top\Sigma_{T,X}w_T}\exp\left\{\frac{[y-(w_S^\top\mu_{T,X}+b_S)]^2}{2w_S^\top\Sigma_{T,X}w_S}-\frac{[y-(w_T^\top\mu_{T,X}+b)]^2}{w_T^\top\Sigma_{T,X}w_T}\right\}.
\end{equation}
More  computations show that
\[\begin{aligned}
        \mathcal{C}_{W}(S,T)&=\left[\mu_{T,Y}-\mu_{S,Y}-\Sigma_{S,YX}\Sigma_{S,X}^{-1}\left(\mu_{T,X}-\mu_{S,X}\right)\right]^2\\
        &+\left(\sqrt{\Sigma_{S,YX}\Sigma_{S,X}^{-1}\Sigma_{T,X}\Sigma_{S,X}^{-1}\Sigma_{S,XY}}-\sqrt{\Sigma_{T,YX}\Sigma_{T,X}^{-1}\Sigma_{T,XY}}\right)^2.
    \end{aligned}\]
One can show that the regret \eqref{eqn:regret} for this transfer leaning problem is given by
\begin{equation}\label{eqn:regret_explicit}
    \RCal(S,T)=\|\Sigma^{\frac{1}{2}}(w_T-w_S)\|_2^2+\left[\mu_{T,Y}-\mu_{S,Y}-\Sigma_{S,YX}\Sigma_{S,X}^{-1}\left(\mu_{T,X}-\mu_{S,X}\right)\right]^2.
\end{equation}
It is easy to verify that 
\begin{equation}\label{eqn:compare_risk_regret}
    \RCal(S,T)=\CCal_{W}(S,T)+2\left(\|\Sigma_{T,X}^{1/2}w_T\|_2\|\Sigma_{T,X}^{1/2}w_S\|_2-\langle\Sigma_{T,X}^{1/2}w_T,\Sigma_{T,X}^{1/2}w_S\rangle\right).
\end{equation}
Now Theorem  \ref{prop: lb} is an immediate consequence of \eqref{eqn:compare_risk_regret} and the Cauchy–Schwartz inequality.
\end{proof}

\begin{rmk}
    Denote the first term in \eqref{eqn:regret_explicit} as $\hat{error}_v(S,T):=\|\Sigma^{\frac{1}{2}}(w_T-w_S)\|_2^2$, and denote the second term in \eqref{eqn:regret_explicit} as $\hat{error}_b(S,T):=\left[\mu_{T,Y}-\mu_{S,Y}-\Sigma_{S,YX}\Sigma_{S,X}^{-1}\left(\mu_{T,X}-\mu_{S,X}\right)\right]^2$. Meanwhile, denote
    \[\CCal_w(S,T)=error_{v,W}(S,T)+error_{b,W}(S,T),\]
    where
    \[\begin{aligned}
    & error_{v,W}(S,T)=\left(\sqrt{\Sigma_{S,YX}\Sigma_{S,X}^{-1}\Sigma_{T,X}\Sigma_{S,X}^{-1}\Sigma_{S,XY}}-\sqrt{\Sigma_{T,YX}\Sigma_{T,X}^{-1}\Sigma_{T,XY}}\right)^2,\\
    & error_{b,W}(S,T)=\left[\mu_{T,Y}-\mu_{S,Y}-\Sigma_{S,YX}\Sigma_{S,X}^{-1}\left(\mu_{T,X}-\mu_{S,X}\right)\right]^2.
    \end{aligned}\]
    That is, transfer risk can be decomposed into two parts, one being  the variance terms $error_{v,W}$ determined by the covariance matrices of the source and target data, and the other being the bias terms $error_{b,W}$ dependent on the difference between the expectations of $\mu_T$ and $\mu_S$.  
    Then,
    \begin{itemize}
        \item A vanishing bias term in transfer risks is equivalent to a vanishing bias term in regret, i.e., $\hat{error}_{b}(S,T)=0\Longleftrightarrow error_{b,W}(S,T)=0$.
        \item A vanishing variance term in transfer risk is necessary for a vanishing variance term in regret, i.e., $\hat{error}_{v}(S,T)=0\Longrightarrow error_{v,W}(S,T)=0$. 
        \item The residual term $2\left(\|\Sigma_{T,X}^{\frac{1}{2}}w_T\|_2\|\Sigma_{T,X}^{\frac{1}{2}}w_S\|_2-\langle\Sigma_{T,X}^{\frac{1}{2}}w_T,\Sigma_{T,X}^{\frac{1}{2}}w_S\rangle\right)$ in \eqref{eqn:compare_risk_regret} depends entirely on the source and target covariance matrices $\Sigma_S$ and $\Sigma_T$is due to the variance term in the learning objective difference. Therefore, when $\CCal_{W}(S,T)=0$, the training process is to reduce the angular distance between $\Sigma_{T,X}^{1/2}w_S$ and $\Sigma_{T,X}^{1/2}w_T$ caused by the discrepancy in these two covariance matrices.
        \item The bias risk component $error_{b,W}(S,T)$ remain strictly positive unless the weighted difference between the expectations $\mu_T$ and $\mu_S$ is $0$.
    \end{itemize}
\end{rmk}

\section{Transfer Learning for  Stock Return Prediction}\label{sec:return_predict}
 Now, through the optimization framework \eqref{eq: doub-trans} and the concept of transfer risk \eqref{eqn:risk}, we will study the stock return prediction problem in this section.  Numerical experiments will show the consistency of the transfer risk with classical statistical metrics. In addition, improved prediction accuracy may be achieved under appropriate  transfer learning over  direct learning. 

\subsection{Prediction problem and setup}\label{sec:return_predict_set_up}
The task is to predict a stock's daily return based on its historical daily return and daily trading volume data.

\paragraph{Data pre-processing via signature transform.}
Denote $\{s_\tau\}_{\tau\geq 0}$ as the historical daily close price of a stock and $\{v_\tau\}_{\tau\geq 0}$ as its historical daily trading volume. To facilitate the subsequent computation, we first apply the log transform to the data, which results in $\{\log(s_\tau)\}_{\tau\geq 0}$ and $\{\log(v_\tau)\}_{\tau\geq 0}$. At each time step $t=0, 1,\dots$, we aim to predict the next daily log return $y_t=\log(s_{t+1})-\log(s_t)$, based on the historical price and volume data $\{\log(s_\tau)\}_{0\leq\tau\leq t}$ and $\{\log(v_\tau)\}_{0\leq\tau\leq t}$.

To adopt  the transfer learning framework, we propose to generate features from the data via the so called \textit{signature transform}, as recalled below; for more details on this topic, see for instance \citep{chen1954iterated,chen1957integration,chen1958integration}, \citep{boedihardjo2016signature}, and \citep{kidger2019deep}.  
\begin{defn}[Signature]\label{defn:signature}
    For a continuous piecewise smooth path $x:[0,T]\to\R^d$, its signature is given by
    \[S(x)=(1,S(x)^{(1)},\dots,S(x)^{(d)},S(x)^{(1,1)},\dots,S(x)^{(i,j)},\dots,S(x)^{(d,d)},\dots),\]
    where for $m=1,2,\dots,$ and $i_1,\dots,i_m\in[d]$,
    \[S(x)^{(i_1,\dots,i_m)}=\int_{0\leq t_1<\dots<t_m\leq T}dx^{i_1}_{t_1}\dots dx^{i_m}_{t_m}.\]
    The truncated signature up to degree $M$ is given by
    \[S_M(x)=(1,S(x)^{(1)},\dots,S(x)^{(\overbrace{n,\dots,n}^{M})}).\]
\end{defn}
It is well known now that the signature of a path generated by a sequence of data essentially determines the path in a computationally  efficient way. Furthermore, its universal nonlinearity property allows for the approximation of  every continuous function of the path  by a linear function of its signature transform. 
Here in the prediction problem,  the features are constructed by the signature transform through the following procedure:
\begin{enumerate}
    \item Fix a time lag $L$ and an order parameter $M$.
    
    \item For each time step $t=0, 1,\dots$, consider the historical price and volume data with time lag $L$: $\{\log(s_\tau)\}_{t-L+1\leq\tau\leq t}$ and $\{\log(v_\tau)\}_{t-L+1\leq\tau\leq t}$.

    \item Consider a three-dimensional path $z_\tau=(\tau, \log(s_\tau), \log(v_\tau))$ for $t-L+1\leq\tau\leq t$, and compute its $M^{\text{th}}$-order truncated signature $x_t=S_M(z)$, as defined in Definition \ref{defn:signature}.

    \item Construct the data set for the prediction problem with the feature vector $x_t$ and the target variable $y_t$: $\{(x_t, y_t)\}_{t\geq 0}$. Each feature, as well as the target variable, will be standardized before inputting into any prediction model.
\end{enumerate}
The time lag parameter $L$ and the order parameter $M$ are the hyper-parameters of the prediction model, and their effects on the prediction accuracy will be studied in the numerical results.

\paragraph{Prediction with Ridge regression via direct learning.} Given the training set $\{(x_t, y_t)\}_{0\leq t\leq T-1}$, and in particular the universal non-linearity property of the signtuare transform, it is natural to consider a linear model $y_t=x_t\cdot\theta+\epsilon$  through the ridge regression. Here  the $L_2$ penalized  addresses the potential multi-collinearity issue of features from the signature transform. 

Given a particular target stock, a direct learning scheme  starts with constructing the signature feature data set $\{(x_t, y_t)\}_{0\leq t\leq T-1}$. Then the entire time horizon will be split into a training period and a testing period. A Ridge regression is then applied to fit the linear model on the train data set, and the performance of the model will be evaluated on the test data set.
More specifically, the Ridge regression finds the estimator $\widehat\theta$ by solving the following $L_2$ penalized least square problem 
\begin{align}\label{eqn:ridge}
\widehat{\theta}:=\argmin_\theta\frac{1}{T}\sum_{t=0}^{T-1}\left(x_t\cdot\theta - y_t\right)^2 + \lambda \left\|\theta\right\|_2^2.
\end{align}
The choice of the hyper-parameter $\lambda>0$ controls the bias-variance tradeoff of the method: a bigger $\lambda$ indicates a lower variance and a higher bias.

Since the data set constructed from one single stock may only contain limited samples, it motivates the idea of boosting the prediction accuracy by transfer learning.

\paragraph{Prediction with Ridge regression via transfer learning.}
Transfer learning allows  one to first pre-train the Ridge regression model on a data set constructed from multiple stocks. That is, consider a target task which is to predict the return of one particular stock, say Apple. One can pre-train a model on a source task which contains several related stocks (such as Google, Amazon, Microsoft) from the similar industry. 
In order to transfer the per-trained model to the target task, one then again fits a new model with the Ridge regression, where the $L_2$ regularization term penalizes on the distance with respect to the pre-trained model.

Here the source task and target task share the same input and output spaces $\mathcal{X}_S=\mathcal{X}_T=\mathbb{R}^d_\text{sig}$ and $\mathcal{Y}_S=\mathcal{Y}_T=\mathbb{R}$, where $d_\text{sig}$ is the dimension of the signature feature and is fully determined by the order $M$. Meanwhile, the admissible sets of source and target models are restricted to all the linear functions from $\mathbb{R}^d_\text{sig}$ to $\mathbb{R}$: $A_S=A_T=\{f:\mathbb{R}^d_\text{sig}\to\mathbb{R}|f(x)=x\cdot\theta\text{ for some }\theta\in\mathbb{R}^d_\text{sig}\}$. Denote the source data set as $\{(x_{S,t}, y_{S,t})\}_{0\leq t\leq T_S-1}$ and the target data set $\{(x_{T,t}, y_{T,t})\}_{0\leq t\leq T_T-1}$. The pre-trained model is first obtained by solving the following optimization problem
\begin{align}\label{eqn:ridge_source}
\widehat{\theta_S}:=\argmin_\theta\frac{1}{T_S}\sum_{t=0}^{T_S-1}\left(x_{S,t}\cdot\theta - y_{S,t}\right)^2 + \lambda_S \left\|\theta\right\|_2^2.
\end{align}
This formulation \eqref{eqn:ridge_source} is equivalent to \eqref{eq: obj-s} by taking the loss function $L_S$ as the square loss plus a regularization term on the $l_2$-norm of the linear parameter $\theta$, while the expectation is taken over the empirical distribution of source samples $\{(x_{S,t}, y_{S,t})\}_{0\leq t\leq T_S-1}$.

Transfer learning is then to find the solution for the following optimization problem
\begin{align}\label{eqn:ridge_target}
\widehat{\theta_T}:=\argmin_\theta\frac{1}{T_T}\sum_{t=0}^{T_T-1}\left(x_{T,t}\cdot\theta - y_{T,t}\right)^2 + \lambda_T \left\|\theta-\widehat{\theta_S}\right\|_2^2.
\end{align}
Here $\lambda_T>0$ in \eqref{eqn:ridge_target} is a hyper-parameter controls the power of the regularization: the higher $\lambda$ is, the closer the transferred model $\widehat{\theta_T}$ will be to the pre-trained model $\widehat{\theta_S}$. 

Note that from the point of view of \eqref{eq: doub-trans}, \eqref{eqn:ridge_target} is equivalent to searching an output transport mapping $T_Y$ over the linear function space $A_T=\{f:\mathbb{R}^d_\text{sig}\to\mathbb{R}|f(x)=x\cdot\theta\text{ for some }\theta\in\mathbb{R}^d_\text{sig}\}$, while the loss function $L_T$ takes the form of the square loss in addition to a regularization on the $l_2$ distance from the source model parameter $\widehat{\theta_S}$.

\subsection{Numerical results}
The data used in the  numerical experiment is the historical daily stock  price and trading volume from the Information Technology sector of the US equity market,  from February 2010 to September 2022. For each experiment, a set of eleven stocks will be first randomly sampled from the Information Technology sector, with the first ten of them served as the source task, and the last one as the target task. Three data sets will be constructed according to the data pre-processing steps discussed in Section \ref{sec:return_predict_set_up}:
\begin{enumerate}
    \item \textbf{Source training data}: which consists of the signature features and log returns of the first ten source stocks, from February 2010 to September 2021.

    \item \textbf{Target training data}: which consists of the signature features and log returns of the last target stock, from February 2010 to September 2021.

    \item \textbf{Target testing data}: which consists of the signature features and log returns of the last target stock, from September 2021 to September 2022.
\end{enumerate}

We compare two approaches
\begin{enumerate}
    \item \textbf{Direct learning}: the model is directly trained on the target training data by the Ridge regression \eqref{eqn:ridge}.
    
    \item \textbf{Transfer learning}: the model is first pre-trained on the source training data by the Ridge regression \eqref{eqn:ridge_source}, and is then retrofit on the target training data by \eqref{eqn:ridge_target}.
\end{enumerate}

The performances are evaluated on the target testing data through three different metrics: mean square error ($\textbf{MSE}$), $\textbf{R}^2$, and correlation between predicted return and actual return (\textbf{Corr}). We also change the hyper-parameters: lag $L$ and order $M$ in the signature-based feature generation steps to study their influences on the performances. The regularization parameters in \eqref{eqn:ridge_source} and \eqref{eqn:ridge_target} are set as $\lambda_S=1.0$ and $\lambda_T=5.0$. The full results are listed in Table \ref{tab:prediction}, and the results are averaged over two hundred random selections of stocks.

\paragraph{Main findings.} In Table \ref{tab:prediction}, the top performer under each metric is marked red. For example, the lowest MSE is obtained when transfer learning  is applied to the prediction task, with $L=2, M=3$. As seen from the table,  transfer learning achieves the best results under all three different metrics, for this task of return prediction. 

Table \ref{tab:prediction_score} 
summarizes the relations between the Wasserstein-based transfer risk $\mathcal{C}(S,T)=\Ef^O_W(\widehat{\theta_S})$ in \eqref{eqn:W_risk}, where $p=2$, and three other standard metrics: MSE, $R^2$, and correlation. More specifically, with $L=5, M=2, \lambda_S=1.0, \lambda_T=5.0$, we randomly choose the source and target stocks as described earlier, apply transfer learning to the prediction problem under these four metrics. Here we record the negations of $\textbf{R}^2$ and $\text{Corr}$ so that lower numbers imply better performances. The correlation matrix, Table \ref{tab:prediction_score}, is computed over two hundred random experiments. It can be observed that the Wasserstein-based transfer gap is positively correlated with other three metrics, especially with the mean square error. This finding suggests that the lower the transfer risk, the better the transfer learning performance.

Finally, it is worth pointing out  the impact of the signature order $M$: in most cases,  larger $M$ ($M\geq4$) does not improve the performance of either transfer learning or direct learning.    This implies that empirically, drift, volatility, and skewness, corresponding to signature features with $M\leq3$,  are more useful in stock return predictions than other higher-order moments. And a large $M$, which may lead to an excessive number of features with possible over-fitting issues, is computationally expensive and unnecessary in the prediction task.

\begin{table}[H]
  \centering
    \begin{tabular}{c|c|c|c|c|c|c|c}
    \toprule
    \toprule
          & \textbf{Algorithm} & \multicolumn{3}{c|}{\textbf{Direct}} & \multicolumn{3}{c}{\textbf{Transfer}} \\
    \midrule
    \textbf{Lag} & \textbf{Order\textbackslash{}Metric} & \multicolumn{1}{c|}{\textbf{MSE}} & \multicolumn{1}{c|}{$\textbf{R}^2$} & \multicolumn{1}{c|}{\textbf{Corr}} & \multicolumn{1}{c|}{\textbf{MSE}} & \multicolumn{1}{c|}{$\textbf{R}^2$} & \multicolumn{1}{c}{\textbf{Corr}} \\
    \midrule
    \multirow{3}[6]{*}{\textbf{2.00}} & \textbf{2.00} & 1.21  & 0.07  & 0.03  & 1.66  & 0.08  & 0.04 \\
\cmidrule{2-8}          & \textbf{3.00} & 1.41  & 0.06  & 0.03  & \textcolor[rgb]{1,0,0}{1.15} & 0.07  & 0.03 \\
\cmidrule{2-8}          & \textbf{4.00} & 1.88  & -0.06 & 0.03  & 1.91  & 0.06  & 0.02 \\
    \midrule
    \multirow{3}[6]{*}{\textbf{3.00}} & \textbf{2.00} & 1.75  & 0.06  & -0.01 &  1.55  & 0.07  & -0.01 \\
\cmidrule{2-8}          & \textbf{3.00} & 1.58  & 0.04  & 0.02  & 1.27  & 0.05  & 0.03 \\
\cmidrule{2-8}          & \textbf{4.00} & 1.72  & -0.14 & 0.00  & 1.70  & -0.07 & 0.01 \\
    \midrule
    \multirow{3}[6]{*}{\textbf{5.00}} & \textbf{2.00} & 1.62  & 0.08  & 0.04  & 1.54  & \textcolor[rgb]{1,0,0}{0.08} & 0.06 \\
\cmidrule{2-8}          & \textbf{3.00} & 1.74  & 0.05  & 0.04  & 2.28  & 0.07  & 0.04 \\
\cmidrule{2-8}          & \textbf{4.00} & 1.63  & -0.13 & 0.02  & 2.05  & 0.05  & 0.06 \\
    \midrule
    \multirow{3}[6]{*}{\textbf{10.00}} & \textbf{2.00} & 1.31  & 0.08  & 0.06  & 1.75  & 0.08  & \textcolor[rgb]{1,0,0}{0.07} \\
\cmidrule{2-8}          & \textbf{3.00} & 1.45  & 0.04  & -0.01 & 1.53  & 0.07  & 0.04 \\
\cmidrule{2-8}          & \textbf{4.00} & 2.48  & -0.14 & 0.07  & 2.84  & -0.01 & 0.06 \\
    \bottomrule
    \bottomrule
    \end{tabular}
    \caption{Prediction performance of direct and transfer learning.}
  \label{tab:prediction}
\end{table}

\begin{table}[H]
  \centering
    \begin{tabular}{c|cccc}
    \toprule
    \toprule
          & \multicolumn{1}{c}{\textbf{MSE}} & \multicolumn{1}{c}{$-\textbf{R}^2$} & \multicolumn{1}{c}{$-$\textbf{Corr}} & \multicolumn{1}{c}{\textbf{Transfer Risk}} \\
    \midrule
    \textbf{MSE} & 1.000 & 0.275 & 0.136 & 0.963 \\
    $-\textbf{R}^2$ & 0.275 & 1.000 & 0.588 & 0.401 \\
    $-$\textbf{Corr} & 0.136 & 0.588 & 1.000 & 0.173 \\
    \textbf{Transfer Risk} & 0.963 & 0.401 & 0.173 & 1.000\\
    \bottomrule
    \bottomrule
    \end{tabular}
  \caption{Correlation matrix of various metrics.}
  \label{tab:prediction_score}
\end{table}

\section{Transfer Learning for Portfolio Optimization}\label{sec:finance_PO}
Portfolio optimization is another natural testing ground for transfer learning: portfolios with assets from newly-emerged markets versus those from more mature markets, portfolios across different industrial sectors, and portfolios under different trading frequencies. In this section, we will test the performance of transfer learning in an unsupervised setting  and its transfer risk  for the  financial portfolio optimization problem. 

In particular, we will consider three types of portfolio transfers: 
\begin{enumerate}
    \item Cross-continent transfer. This refers to the case when one transfers a portfolio from the equity market in one country to the equity market in another, e.g., from the US equity market to the Brazil equity market. In general, the source market has more historical data or more diverse stocks than the target market, which may provide the target market with a robust pre-trained portfolio. The study of cross-continent transfer aims to understand how continental discrepancy affects the performance of transfer learning. 

    \item Cross-sector transfer. This refers to the case when one transfers a portfolio from one sector of a market to another sector, e.g., from the Information Technology sector to the Health Care sector in the US equity market. The study of cross-sector transfer aims to understand correlations between various sectors in the market, and how correlations between sectors affect the performance of transfer learning.

    \item Cross-frequency transfer. This refers to the case when one transfers a portfolio constructed under one trading frequency to another trading frequency, e.g., from low-frequency trading to mid or high-frequency trading. The study of cross-frequency transfer aims to explore the possibility of transferring the portfolio across different trading frequencies, which may be relevant for institutional investors.
\end{enumerate}


\subsection{Portfolio Optimization Problem Set-up}\label{subsec:port_opt_setup}
\subsubsection{Portfolio Optimization Based on Sharpe Ratio} 
Consider a capital market consisting of $d$ assets whose annualized returns are captured by the random vector $r = (r_1, ..., r_d)^\top \sim \mathbb{P}.$ A portfolio allocation vector $\phi = (\phi_1, ..., \phi_d)^\top$ is a $d$-dimensional vector in the unit simplex $\mathbb{X} := \{\phi \in \mathbb{R}^d_+: \sum_{j=1}^d \phi_j = 1 \}$ with $\phi_i$ percentage of the available capital invested in asset $i$ for each $i=1,...,d$. The annualized return of a portfolio $\phi$ is given by $\phi^\top r.$ 

We aim to find the optimal portfolio with the highest Sharpe ratio by solving the following optimization problem:    
\begin{equation}\label{eqn:portfolio_opt}
    \widehat{\phi}=\argmax_{\phi \in \mathbb{X}}\frac{\mathbb{E}^{\mathbb{P}}[\phi^\top r]}{\text{Std}(\phi^\top r)}=\argmax_{\phi \in \mathbb{X}}\frac{\mu_\mathbb{P}^\top\phi}{\sqrt{\phi^\top\Sigma_\mathbb{P}\phi}},
\end{equation}
where $\mu_\mathbb{P}$ is the expectation and $\Sigma_\mathbb{P}$ is the covariance matrix  of the return $r$. Empirically, $\mu_\mathbb{P}$ and $\Sigma_\mathbb{P}$ are estimated from the historical return.

In many cases, such as managing portfolios in new emerging markets, there are limited data for directly estimating $\mu_\mathbb{P}$ and $\Sigma_\mathbb{P}$, which may result in large estimation error and lead to a non-robust portfolio. We will show here that transfer learning could be a natural and viable framework to resolve this problem.

\subsubsection{Portfolio Optimization with Transfer Learning}
In this example, the source task and target task share the same input and output spaces $\mathcal{X}_S=\mathcal{X}_T=\mathbb{R}^d$ and $\mathcal{Y}_S=\mathcal{Y}_T=\mathbb{R}$, where $d$ is the number of assets in the portfolio. Note that different from the previous return prediction example, there is no sample from the output space that is explicitly observed. Instead, one can only get historical stock returns as input data. Meanwhile, a portfolio vector $\phi$ can be viewed as a linear mapping from $\mathbb{R}^d$ to $\mathbb{R}$, taking a $d$-dimensional stock return to a portfolio return. More specifically, the admissible sets of source and target models are restricted to: $A_S=A_T=\{f:\mathbb{R}^d\to\mathbb{R}|f(r)=\phi^\top r\text{ for some }\phi\in\mathbb{X}\}$. In addition, for the source task, the loss functional $\mathcal{L}_S$ in \eqref{eq: obj-s} is set to be the negative Sharpe ratio, and consequently, the source task is to solve following the optimization problem:
\begin{equation}\label{eqn:portfolio_opt_source}
    \widehat{\phi_S}=\argmax_{\phi \in \mathbb{X}}\frac{\mu_S^\top\phi}{\sqrt{\phi^\top\Sigma_S\phi}},
\end{equation}
where $\mu_S$ and $\Sigma_S$ are the mean and covariance estimations from the source data set. 

Then, the portfolio is transferred to the target task by solving the following optimization with a $L_2$ regularization term penalizing the distance between the pre-trained portfolio and the transferred portfolio:
\begin{equation}\label{eqn:portfolio_opt_target}
    \widehat{\phi_T}=\argmax_{\phi \in \mathbb{X}}\frac{\mu_T^\top\phi}{\sqrt{\phi^\top\Sigma_T\phi}}-\lambda\left\|\widehat{\phi_S}-\phi\right\|^2_2.
\end{equation}
Here $\mu_T$ and $\Sigma_T$ are the mean and covariance estimations from the target data, and $\lambda>0$ is a hyper-parameter controls the power of the regularization: the higher $\lambda$ is, the closer the transferred portfolio $\widehat{\phi_T}$ will be to the pre-trained portfolio $\widehat{\phi_S}$. 

On one hand, by adding the penalty term \(-\lambda\|\hat{\phi_S}-\phi\|_2^2\) in \eqref{eqn:portfolio_opt_target}, we add some ``supervised learning'' flavor to this unsupervised learning problem: provided that transfer risk is low, optimal target portfolio should be comparable with the source one. On the other hand, from the viewpoint of Remark \ref{rmk:tl-unsup}, \eqref{eqn:portfolio_opt_target} is equivalent to searching an output transport mapping $T^Y$ over the linear function space $\T^Y=A_T=\{f:\mathbb{R}^d\to\mathbb{R}|f(r)=\phi^\top r\text{ for some }\phi\in\mathbb{X}\}$, while the loss functional $\mathcal{L}_T$ in \eqref{eq: doub-trans} takes the form of negative Sharpe ratio in addition to a regularization on the $l_2$ distance from the source model $\widehat{\phi_S}$. The entire procedure of the portfolio transfer is summarized in Figure \ref{fig:procedure}.
\begin{figure}
    \centering
    \includegraphics[width=0.45\textwidth]{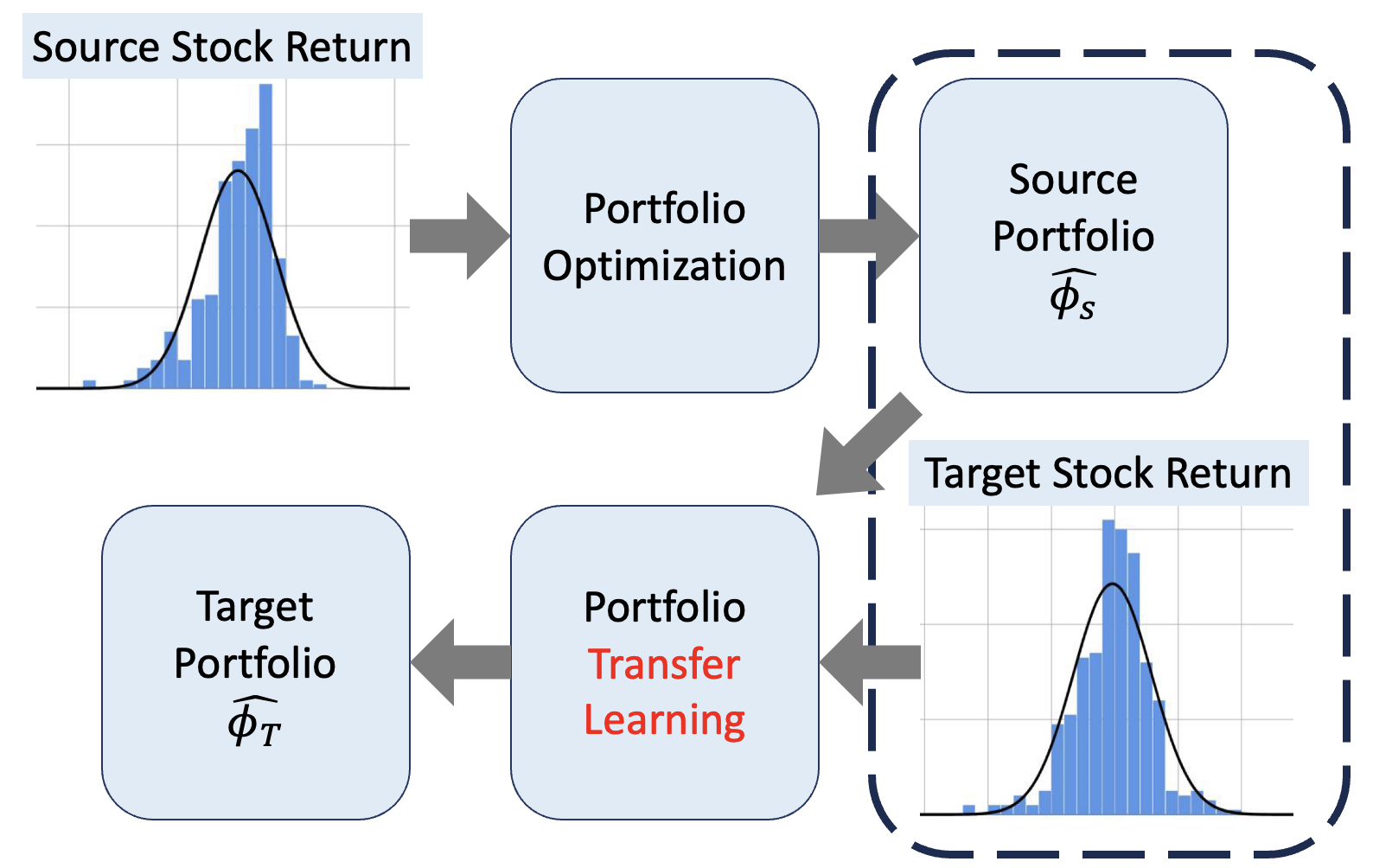}
    \caption{Procedure of portfolio transfer.}
    \label{fig:procedure}
\end{figure}





\subsubsection{Computation of Transfer risk}\label{subsec:port_opt_tf_risk}
Under the context of transfer learning for portfolio optimization, the concept of transfer risk encompasses two aspects, namely the ``quality'' and the ``relevance'' of the chosen source portfolio. 
That is, the transfer risk $\mathcal{C}_{PO}(S,T)$ is expressed as
\begin{equation}
\label{eqn:po-tr}
\mathcal{C}_{PO}(S,T)=R_1+R_2.
\end{equation}
Here $R_1$  concerns with the performance of the source portfolio, and is defined as
\begin{equation}
\label{eqn:R1}
R_1 = \left(\frac{\mu_S^\top\widehat{\phi_S}}{\sqrt{\widehat{\phi_S}^\top\Sigma_S\widehat{\phi_S}}}\right)^{-1}.
\end{equation}
This expression is inversely proportional to the Sharpe ratio of the source task, and it serves as a measure of risk associated with selecting a source task that exhibits poor portfolio performance.

The second component $R_2$ measures the similarity between the source and target portfolios in terms of their data distributions. More specifically, it approximates the return distributions of the source and target tasks using mean and covariance estimates $\left(\mu_S, \Sigma_S\right)$ and $\left(\mu_T, \Sigma_T\right)$, respectively, and choose probability distributions \(\mathcal{M})_S(\mu_S,\Sigma_S)\) and \(\mathcal{M}_T(\mu_T,\Sigma_T)\) that appropriately fit the portfolio data with matching first two moments, respectively. Consequently, $R_2$ is defined as the Wasserstein-2 distance between these two distributions:
\begin{equation}
\label{eqn:R2}
R_2=\mathcal{W}_2\left(\mathcal{M}\left(\mu_S, \Sigma_S\right), \mathcal{M}\left(\mu_T, \Sigma_T\right)\right).
\end{equation}
This definition can be essentially viewed as computing the input transport risk for the identity input mapping $T_0^X=id_{\mathcal{X}}$.

\subsection{Experimental Settings and Numerical Results}
Throughout the experiments, we test the performance of transfer learning and transfer risk under various portfolio optimization problems, including cross-continent transfer, cross-sector transfer, and cross-frequency transfer. For the ease of illustration, we adopt multivariate Gaussian distributions for source and target data, i.e., \(\mathcal{M}\left(\mu_{\cdot}, \Sigma_{\cdot}\right)=\mathcal{N}\left(\mu_{\cdot}, \Sigma_{\cdot}\right)\); the possibility of fitting the data with different choices of probability distributions including sub-Gaussian, and heavy-tail distributions will be explored in future works.

\subsubsection{Cross-Continent Transfer} 
In these numerical experiments, the source market is defined as the US equity market, while the target markets are chosen to be United Kingdom, Brazil, Germany, and Singapore, in four separate experiments respectively. Our findings indicate that transfer learning is more likely to outperform direct learning in European markets, such as Germany, while its performance is comparatively worse in the Brazil market. Notably, we observe a strong correlation between transfer risk and the transfer learning performance across all the different markets.

Given a target market, we first select out the top ten stocks with the largest market capitals ($d=10$) as the class of target assets. Then, ten stocks will be randomly selected from the S\&P500 component stocks as the class of source assets. Three data sets will be constructed accordingly:
\begin{enumerate}
    \item \textbf{Source training data}: it consists of the daily returns of ten source assets, from February 2000 to February 2020.

    \item \textbf{Target training data}: it consists of the daily returns of ten target assets, from February 2015 to February 2020.

    \item \textbf{Target testing data}: it consists of the daily returns of ten target assets, from February 2020 to September 2021.
\end{enumerate}

We compare direct learning with transfer learning: for direct learning, the portfolio is directly learned by solving \eqref{eqn:portfolio_opt}, with mean and covariance estimated from target training data; for transfer learning, the portfolio is first pre-trained on the source training data by solving \eqref{eqn:portfolio_opt_source}, then fine-tuned on the target training data by solving \eqref{eqn:portfolio_opt_target}. Finally, the performances of those methods are evaluated through their Sharpe ratios on the target testing data. The regularization parameter $\lambda$ in \eqref{eqn:portfolio_opt_target} is set to be $0.2$. Meanwhile, we also compute the transfer risk following \eqref{eqn:po-tr} in Section \ref{subsec:port_opt_setup}, using the source training data and target testing data.

For each target market, the results across one thousand random experiments (randomness in selections of source assets) are plotted in Figure \ref{fig:cross_continent}. 

Across those four markets, a consistent pattern is observed: the transfer risk is significantly correlated with the Sharpe ratio of the transferred portfolio (with correlation around -0.60). This observation supports the idea of using transfer risk as a measurement for the transferability of a task. Meanwhile, the performance of transfer learning is compared with that of direct learning in Figure \ref{fig:cross_continent}. For each target market, the dashed green line indicates the direct learning performance. The blue dots above the green line represent transfer learning tasks that outperform the direct learning. Note that for all four target markets, there are a significant amount of transfer learning tasks outshining the direct learning, especially for those tasks achieving low transfer risk.

As shown in Figure \ref{fig:cross_continent}, transfer learning is more likely to outperform direct learning in European markets, such as Germany, while its performance is worse in the Brazil market. A possible interpretation to this finding is that European markets have tighter connections and higher similarities with the US market, which help to boost the performance of transfer learning.

\subsubsection{Cross-Sector Transfer} 
In this numerical experiment, we focus on transfer learning among ten different sectors in the US equity market: Communication Services, Consumer Discretionary, Energy, Financials, Health Care, Industrials, Information Technology, Materials, Real Estate, and Utilities. We conduct two separate experiments, one with S\&P500 stocks and the other with non-S\&P500 stocks, to compare the differences in transfer learning performance and transfer risk. The analysis reveals that transfer risks in Health Care and Information Technology sectors display large negative correlations with transfer learning outcomes, effectively characterizing transferability in these sectors. In contrast, correlations are not significant for Utilities and Real Estate, potentially due to factors not fully captured by transfer risk. Additionally, for some sectors such as Energy, the correlations become more significant within non-S\&P500 stocks than S\&P500 stocks.



\begin{figure}[H]
     \centering
     \begin{subfigure}[b]{0.4\textwidth}
         \centering
         \includegraphics[width=\textwidth]{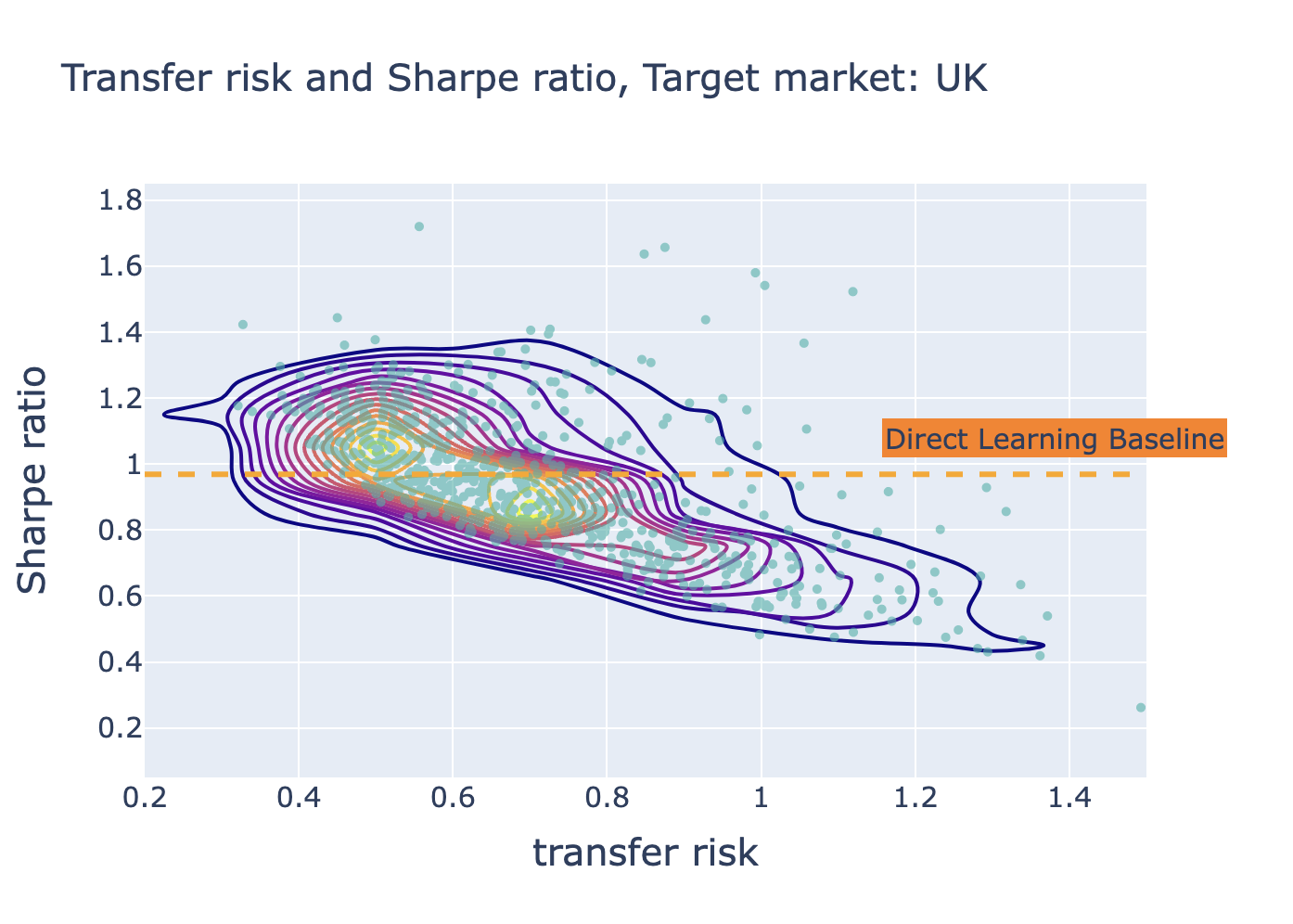}
         \caption{US$\to$UK, correlation=-0.66}
         \label{fig:US-UK}
     \end{subfigure}
     ~
     \begin{subfigure}[b]{0.4\textwidth}
         \centering
         \includegraphics[width=\textwidth]{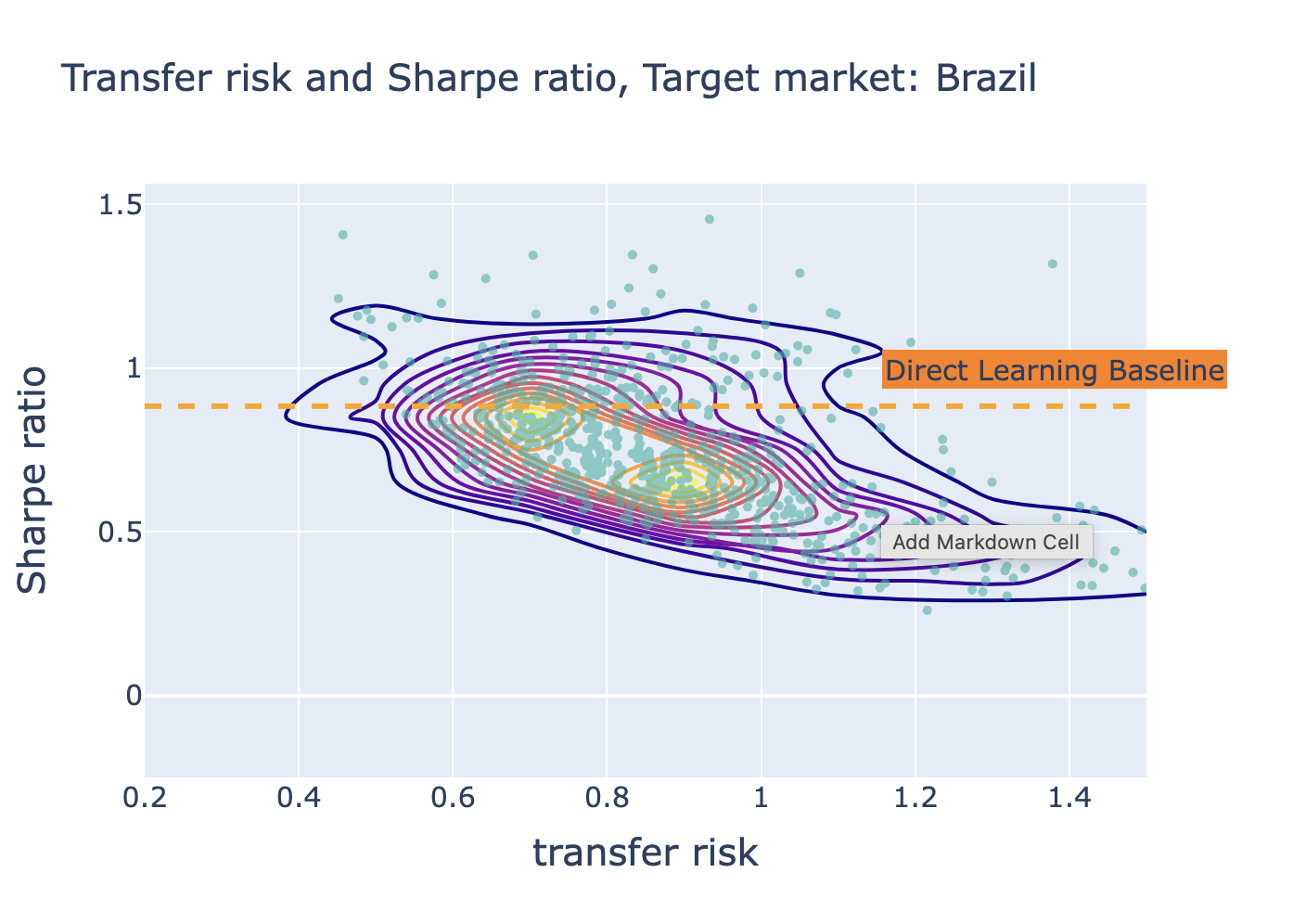}
         \caption{US$\to$Brazil, correlation=-0.67}
         \label{fig:US-Brazil}
     \end{subfigure}
     \\
     \begin{subfigure}[b]{0.4\textwidth}
         \centering
         \includegraphics[width=\textwidth]{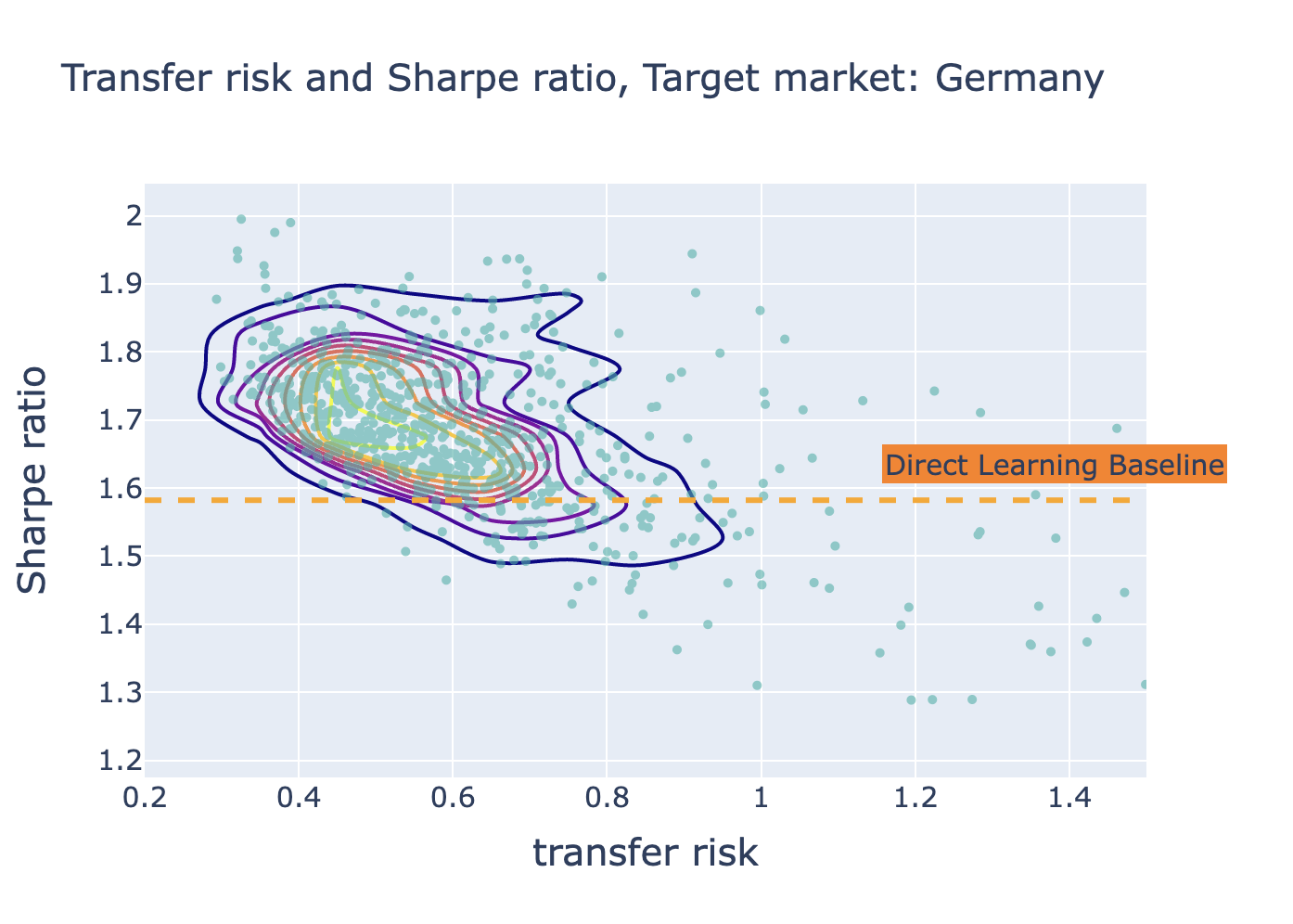}
         \caption{US$\to$Germany, correlation=-0.64}
         \label{fig:US-Germany}
     \end{subfigure}
     ~
     \begin{subfigure}[b]{0.4\textwidth}
         \centering
         \includegraphics[width=\textwidth]{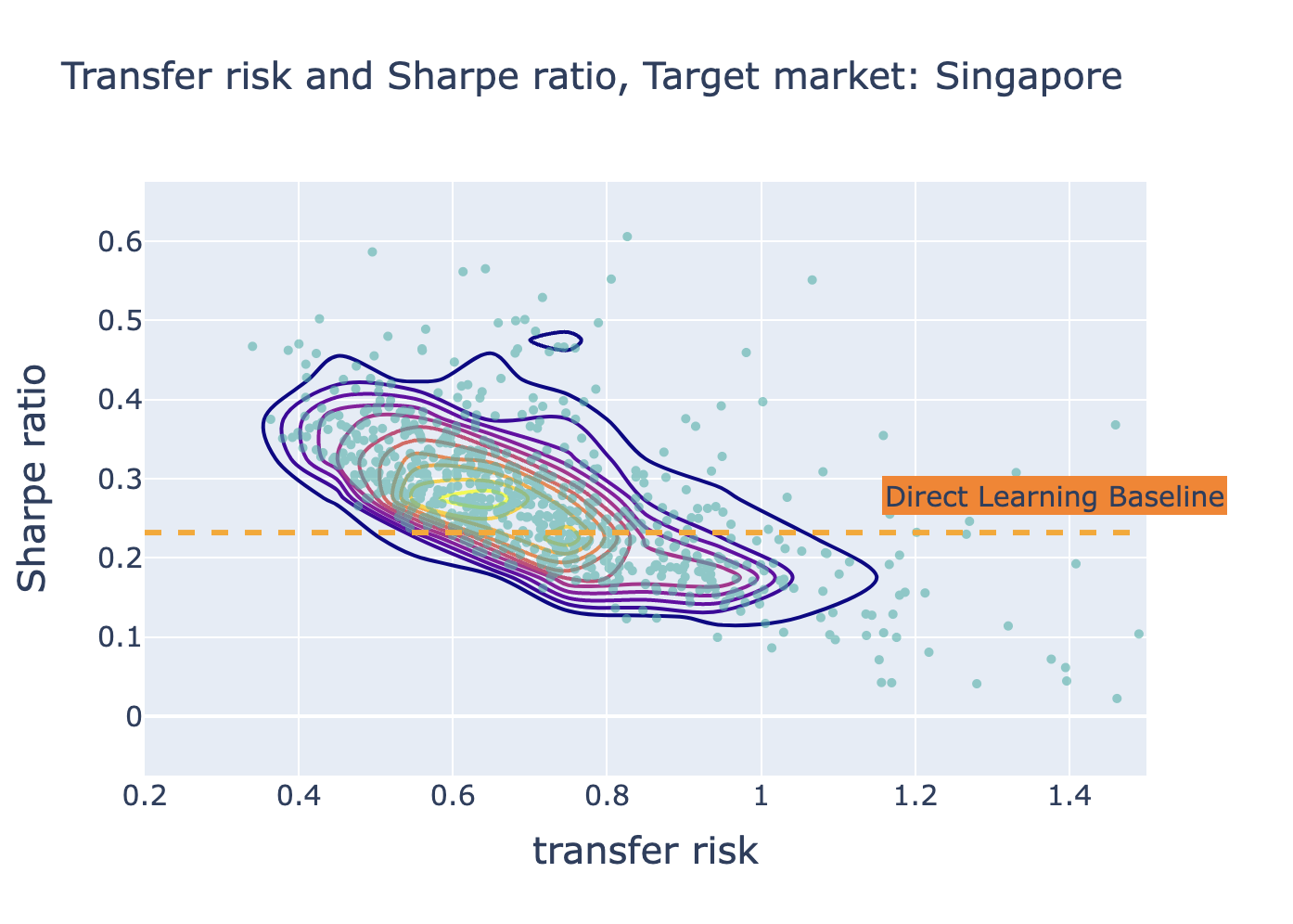}
         \caption{US$\to$Singapore, correlation=-0.62}
         \label{fig:US-Singapore}
     \end{subfigure}
     \caption{Sharpe ratio and transfer risk when transferring from the US market to other markets.}
     \label{fig:cross_continent}
\end{figure}

More specifically, given a source sector and a target sector, we first randomly sample ten stocks ($d=10$) from each sector, as the source asset class and target asset class. Then, three data sets will be constructed accordingly:

\begin{enumerate}
    \item \textbf{Source training data}: it consists of the daily returns of ten stocks from a given source sector, from February 2000 to February 2020.

    \item \textbf{Target training data}: it consists of the daily returns of ten stocks from a given target sector, from February 2015 to February 2020.

    \item \textbf{Target testing data}: it consists of the daily returns of ten stocks from a given target sector, from February 2020 to September 2021.
\end{enumerate}

The transfer learning scheme applied in the experiments is same as before: the portfolio is first pre-trained on the source training data by solving \eqref{eqn:portfolio_opt_source}, then fine-tuned on the target training data by solving \eqref{eqn:portfolio_opt_target}. Finally, the performance of the portfolio is evaluated through its Sharpe ratio on the target testing data. The regularization parameter $\lambda$ in \eqref{eqn:portfolio_opt_target} is set to be $0.2$. Meanwhile, the computation of transfer risk follows \eqref{eqn:po-tr} in Section \ref{subsec:port_opt_setup}, using the source training data and target testing data. 

For each source-target sector pair (in total $10\times10=100$ pairs), five hundred random experiments with different stock selections are conducted, and we record the average Sharpe ratio and average transfer risk of those random experiments. 


Figure \ref{fig:sector_sp500_health} shows the relation between (average) Sharpe ratios and (average) scores when transferring portfolios from various source sectors to the target Health Care sector. In general, the negative correlation between Sharpe ratios and scores is observed: when the target sector is fixed (Health Care in this example), transferring a portfolio from a source sector with lower transfer risk is more likely to achieve a higher Sharpe ratio. In particular, Information Technology, Health Care, and Consumer Discretionary are desirable source sectors when the target sector is Health Care, while Energy is not a suitable choice.
\begin{figure}[!ht]
    \centering
    \includegraphics[width=.75\textwidth]{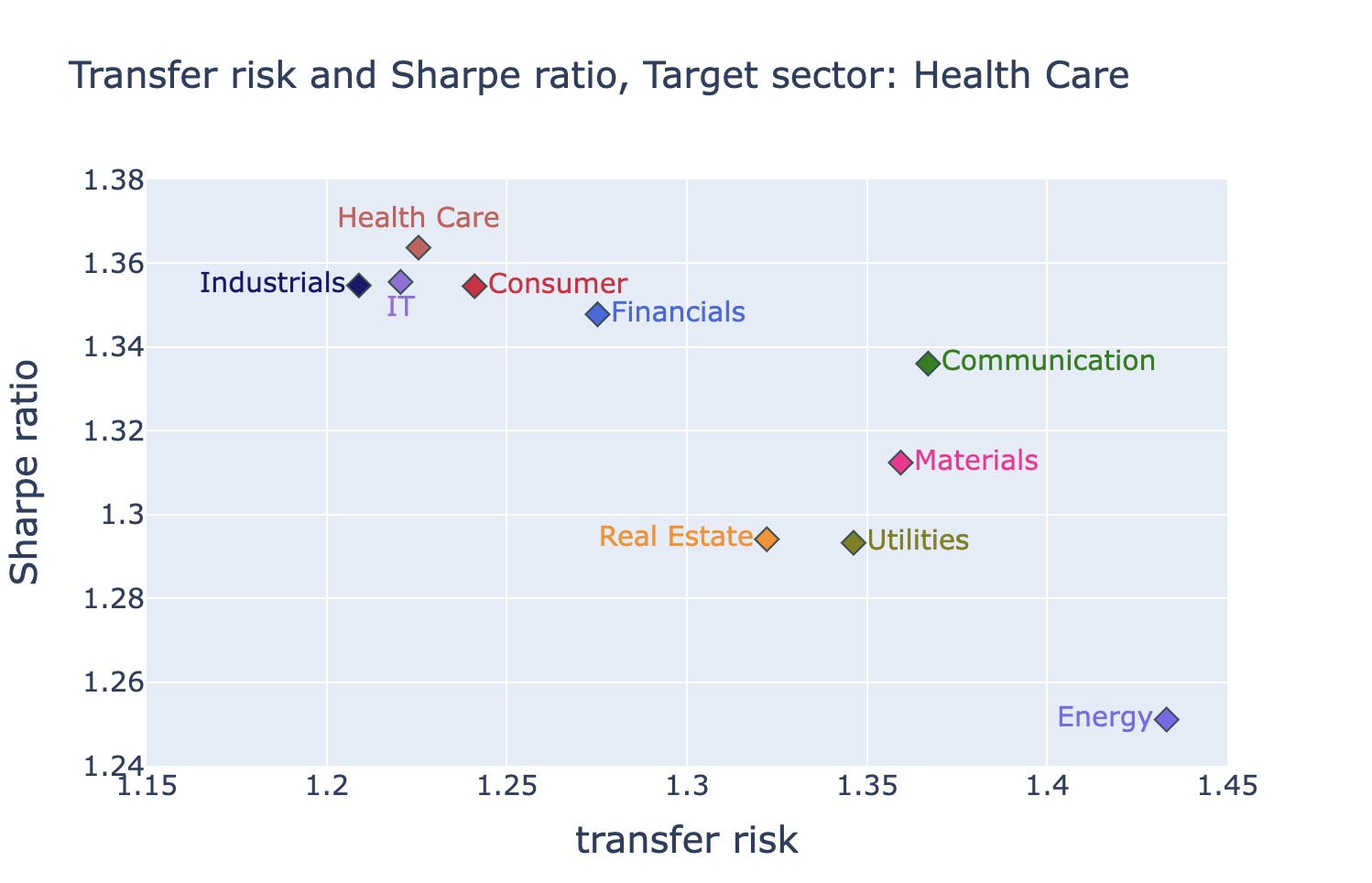}
    \caption{Transfer risk and Sharpe ratio, transferring to Health Care sector in S\&P500 stocks.}
    \label{fig:sector_sp500_health}
\end{figure}

Table \ref{tab:cross_sector_corr} records the correlation between Sharpe ratios and transfer risks when transferring portfolios from various source sectors to each target sector. Lower correlation implies that the target sector may benefit more from transfer learning and the transferability is better captured by transfer risk. 
Table \ref{tab:sp500_corr} is from the experiments using S\&P500 source stocks, while Table \ref{tab:nonsp500_corr} is from the experiments using non-S\&P500 source stocks.


A number of  patterns are observed from above. For sectors such as Health Care and Information Technology, large negative correlations are revealed in Table \ref{tab:cross_sector_corr}, regardless of whether  the stocks are chosen from S\&P500 or not. This implies that the transfer risk appropriately encodes the statistical property and characterizes the transferability of portfolios in those sectors. In constrast, for sectors such as Utilities and Real Estate, the correlations shown in Table \ref{tab:cross_sector_corr} are not significant, regardless of whether the stocks are chosen from S\&P500 or not. This may be due to  the fact that companies in Utilities and Real Estate sectors tend to be affected by underlying spatial factors and also changes in regulatory policies. Those aspects may not be fully captured by the transfer risk. In addition, for Energy sector, the correlation is more significant when considering non-S\&P500 stocks. This may be due to the industry concentration of Energy sector in S\&P500 Index: Energy sector in S\&P500 Index is highly concentrated in a few large companies, and the portfolio's performance is largely driven by some company-specific factors which the transfer risk fails to capture. The effect of industry concentration dwindles when non-S\&P500 stocks are considered.

\begin{table}[H]
    \begin{subtable}[H]{0.45\textwidth}
        \centering
        \begin{tabular}{cc}
            \toprule
            \toprule
            \textbf{Target Sector} & \textbf{Correlation} \\
            \midrule
            \textbf{Health Care} & -0.88 \\
            \textbf{Communication} & -0.82 \\
            \textbf{Materials} & -0.78 \\
            \textbf{IT} & -0.71 \\
            \textbf{Financials} & -0.60 \\
            \textbf{Consumer} & -0.56 \\
            \textbf{Industrials} & -0.54 \\
            \textbf{Real Estate} & -0.37 \\
            \textbf{Utilities} & -0.20 \\
            \textbf{Energy} & 0.04 \\
            \bottomrule
            \bottomrule
        \end{tabular}
       \caption{S\&P500 stocks.}
       \label{tab:sp500_corr}
    \end{subtable}
    \hfill
    \begin{subtable}[H]{0.45\textwidth}
        \centering
        \begin{tabular}{cc}
            \toprule
            \toprule
            \textbf{Target Sector} & \textbf{Correlation} \\
            \midrule
            \textbf{IT} & -0.67 \\
            \textbf{Materials} & -0.49 \\
            \textbf{Health Care} & -0.47 \\
            \textbf{Communication} & -0.46 \\
            \textbf{Energy} & -0.40 \\
            \textbf{Industrials} & -0.30 \\
            \textbf{Consumer} & -0.28 \\
            \textbf{Real Estate} & -0.14 \\
            \textbf{Financials} & 0.10 \\
            \textbf{Utilities} & 0.17 \\
            \bottomrule
            \bottomrule
        \end{tabular}
        \caption{Non-S\&P500 stocks.}
        \label{tab:nonsp500_corr}
     \end{subtable}
     \caption{Correlation between risk and Sharpe ratio for transfer from other sectors to the target.}
     \label{tab:cross_sector_corr}
\end{table}


\subsubsection{Cross-Frequency Transfer} 
In the following numerical experiments, we focus on transfer learning between different trading frequencies, ranging from mid-frequency to low-frequency: 1-minute, 5-minute, 10-minute, 30-minute, 65-minute, 130-minute and 1-day. The findings demonstrate that transferring a low-frequency portfolio (1-day) to higher frequencies results in relatively high transfer risks and poor transfer learning performances. This discrepancy arises from the distinct statistical properties of intraday price movements in mid/high-frequency trading compared to cross-day price movements in low-frequency trading, affecting the transfer process. Conversely, within the mid/high-frequency regime (1-minute to 130-minute), the study reveals that 65-minute and 130-minute frequencies serve as better candidates for the source frequency due to more robust mean and covariance estimations. Consequently, these frequencies lead to improved transfer learning performance after fine-tuning.

More specifically, given a source frequency and a target frequency, we first randomly sample ten stocks ($d=10$) from the fifty largest US companies by market capitalization. Then, three data sets will be constructed accordingly:

\begin{enumerate}
    \item \textbf{Source training data}: it consists of the returns of ten stocks, sampled under a given source frequency, from February 2016 to September 2019.

    \item \textbf{Target training data}: it consists of the returns of ten stocks, sampled under a given target frequency, from February 2016 to September 2019.

    \item \textbf{Target testing data}: it consists of the returns of ten stocks, sampled under a given target frequency, from September 2019 to February 2020.
\end{enumerate}

The transfer learning scheme applied in the experiments is the same as before: the portfolio is first pre-trained by solving \eqref{eqn:portfolio_opt_source} with the mean and covariance estimated from the source-frequency training data, then fine-tuned by solving \eqref{eqn:portfolio_opt_target} with the mean and covariance estimated from the target-frequency training data. Meanwhile, following the usual setting in mid/high-frequency trading, we assume that over-night holding is not allowed for trading frequencies ranging from 1-minute to 130-minute. More specifically, when over-night holding not is allowed, the price movement after the market close and before the market open will not be included in the mean and covariance estimation. Finally, the performance of the portfolio is evaluated through its Sharpe ratio on the target-frequency testing data. The regularization parameter $\lambda$ in \eqref{eqn:portfolio_opt_target} is set to be $0.2$. Meanwhile, the computation of transfer risk follows the approach described in Section \ref{subsec:port_opt_setup}, using the source-frequency training data and target-frequency testing data. 

For each source-target frequency pair (in total $7\times17=49$ pairs), two hundred random experiments with different stock selections are conducted, and we record the average Sharpe ratio and average transfer risk of those random experiments. The results are presented in Figure \ref{fig:freq_intraday_130min}, Figure \ref{fig:cross_freq_intraday} and Table \ref{tab:cross_freq_corr}. 

For example, Figure \ref{fig:freq_intraday_130min} shows the relation between (average) Sharpe ratios and (average) transfer risks when transferring portfolios from various source frequencies to the target frequency of 130-minute. In general, the negative correlation between Sharpe ratios and scores is observed: when the target frequency is fixed (130-minute in this example), transferring a portfolio from a source frequency with lower transfer risk corresponds to a higher Sharpe ratio. In particular, source frequencies such as 130-minute and 65-minute, which are closer to the 130-minute target frequency, are desirable source tasks, while 1-minute or 1-day is less suitable.

To see more clearly the relation between different frequencies, in Figure \ref{fig:cross_freq_intraday}, we plot the heat maps of transfer risks and Sharpe ratios when transferring across all frequencies. Here the transfer risks and Sharpe ratios are again rescaled linearly, so that for each target frequency, the values will range from zero to one. Meanwhile, Table \ref{tab:cross_freq_corr} records the correlation between Sharpe ratios and transfer risks when transferring portfolios from various source frequencies to each target frequency, following the same setting as Figure \ref{fig:cross_freq_intraday}.

From Figure \ref{fig:cross_freq_intraday}, it is observed that the transfer risks from 1-day frequency to other higher frequencies are relatively high, resulting in poor transfer learning performances as well. This demonstrates a natural discrepancy between low-frequency trading and mid/high-frequency trading: mid/high-frequency trading aims to capture intraday stock price movements by not allowing over-night holding, while low-frequency trading intends to capture price movements across trading days. Consequently, the difference between the underlying statistical properties of intraday price movements and cross-day price movements hurts the performance of transferring a low-frequency portfolio to a mid/high-frequency portfolio. 

Meanwhile, for transfer learning inside the mid/high-frequency regime (1-minute to 130-minute), the results in Figure \ref{fig:cross_freq_intraday} reveal that 65-minute and 130-minute are more appropriate candidates for the source frequency, since they lead to much better transfer learning performance, compared to other higher source frequencies. This may be due to the fact that under 65-minute and 130-minute frequencies, the mean and covariance estimations are more robust, hence resulting in a robust source portfolio which performs well after fine-tuning.

\begin{figure}[H]
    \centering
    \includegraphics[width=0.75\textwidth]{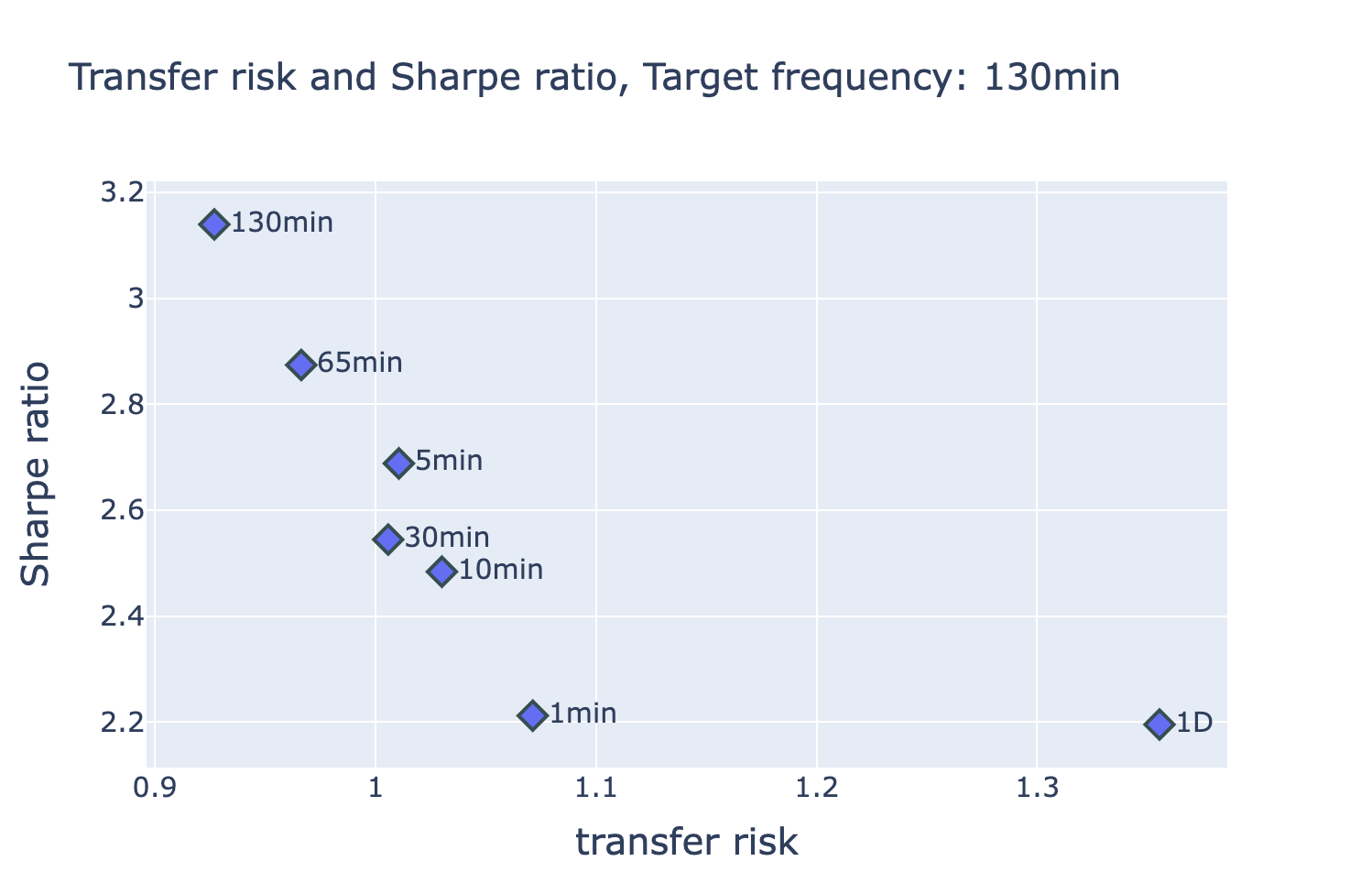}
    \caption{Transfer risk and Sharpe ratio, transferring to 130-minute frequency.}
    \label{fig:freq_intraday_130min}
\end{figure}


\begin{figure}[H]
     \centering
     \begin{subfigure}[b]{0.49\textwidth}
         \centering
         \includegraphics[width=\textwidth]{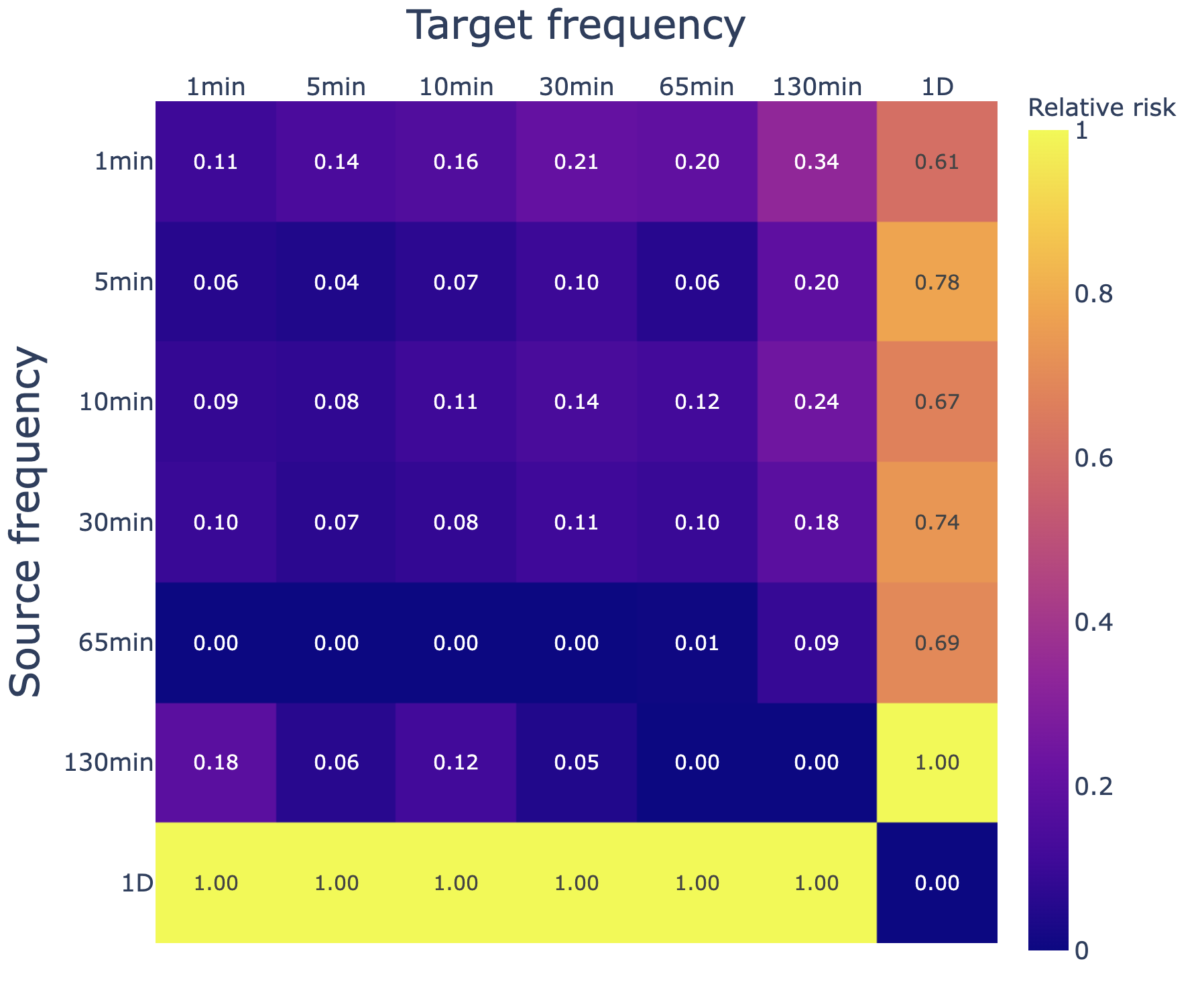}
         \caption{Risk for cross-frequency transfer.}
         \label{fig:score_intraday}
     \end{subfigure}
     \hfill
     \begin{subfigure}[b]{0.49\textwidth}
         \centering
         \includegraphics[width=\textwidth]{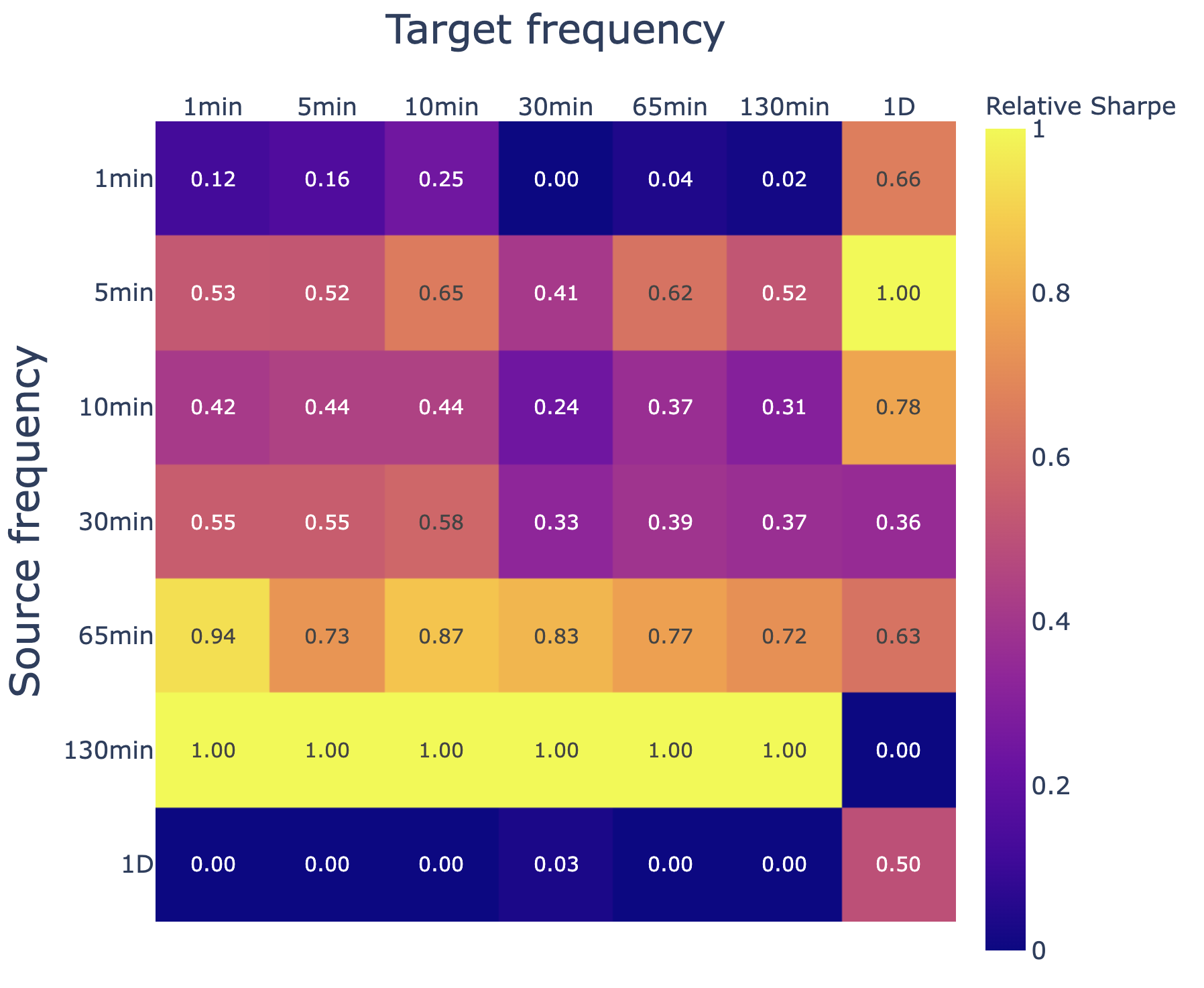}
         \caption{Sharpe ratio for cross-frequency transfer.}
         \label{fig:sharpe_intraday}
     \end{subfigure}
     \caption{Relative risk and Sharpe ratio for transfer across frequencies.}
     \label{fig:cross_freq_intraday}
\end{figure}


\begin{table}[H]
    \centering
    \begin{tabular}{cc}
        \toprule
        \toprule
        \textbf{Target Frequency} & \textbf{Correlation} \\
        \midrule
        \textbf{1 min} & -0.59 \\
        \textbf{5 min} & -0.70 \\
        \textbf{10 min} & -0.74 \\
        \textbf{30 min} & -0.59 \\
        \textbf{65 min} & -0.70 \\
        \textbf{130 min} & -0.76 \\
        \textbf{1 Day} & -0.19 \\
        \bottomrule
        \bottomrule
    \end{tabular}
    \caption{Correlation between risk and Sharpe ratio for transfer from other frequency to the target.}
    \label{tab:cross_freq_corr}
\end{table}

\section{Conclusion}
This paper proposes a novel concept of transfer risk is introduced. Through extensive numerical experiments on classical  financial problems, we show that prior to starting a full-scale transfer learning scheme, transfer risk is an easy-to-compute and viable quantity as a prior estimate of the final learning outcome.

\bibliographystyle{apalike}
\bibliography{refs}

\end{document}